%% file: Homonymous-ICDCS-2012-V9-full.tex
\documentclass[10pt, conference, compsocconf]{IEEEtran}
\usepackage{epic,eepic,amsmath,latexsym,amssymb,color}
\usepackage{ifthen,graphics,epsfig}
\usepackage[english]{babel}
\usepackage{times}
\usepackage{amssymb,amsthm,amscd}
\usepackage{graphicx}

\usepackage{times}
\usepackage{amssymb}
\usepackage{listings}

\newtheorem{definition}{Definition}

\newtheorem{theorem}{Theorem}
\newtheorem{lemma}{Lemma}
\newtheorem{corollary}{Corollary}  
  
\newcommand{\toto}{xxx}

\newcommand{\tightparagraph}[1]{\smallskip \noindent \textbf{#1} \hspace{1ex}}

\newcounter{linecounter}
\newcommand{\linenumbering}{\ifthenelse{\value{linecounter}<10}{(0\arabic{linecounter})}{(\arabic{linecounter})}}

\renewcommand{\thelinecounter}{\ifnum \value{linecounter} > 9\else 0\fi \arabic{linecounter}}
\newcommand{\OO}{{\Omega}}
\newcommand{\HO}{{H\Omega}}

\newcommand{\AO}{{A\Omega}}
\newcommand{\HS}{{H\Sigma}}
\newcommand{\AS}{{A\Sigma}}
\newcommand{\HP}{{ \Diamond H\overline{P}}}
\newcommand{\AP}{{ \mathit{AP}}}
\newcommand{\NAP}{{ \overline{\mathit{AP}}}}
\newcommand{\C}{{\mathit{Correct}}}
\newcommand{\trusted}{{\mathit{h\_trusted}}}
\newcommand{\leader}{{\mathit{h\_leader}}}

\newcommand{\multiplicity}{{\mathit{h\_multiplicity}}}

\newcommand{\quora}{{\mathit{h\_quora}}}
\newcommand{\labels}{{\mathit{h\_labels}}}

{\title{\bf  Failure Detectors in Homonymous Distributed Systems\\ 
(with an Application to Consensus)}}
\vspace{-0.5cm}
{\author{
           Sergio   {\sc Ar\'evalo}$^{\star}$ 
           ~ Antonio  {\sc Fern\'andez Anta}$^{\star\star}$ 
           ~ Damien   {\sc Imbs}$^{\ddag}$ \\
           ~ Ernesto  {\sc Jim\'enez}$^{\star}$
           ~ Michel   {\sc Raynal}$^{{\dag},{\ddag}}$ \\~\\
$^{\star}$ EUI, Universidad Polit\'ecnica de Madrid, 28031 Madrid, Spain\\
$^{\star\star}$ Institute IMDEA Networks, 28918 Madrid, Spain\\
$^{\dag}$  Institut Universitaire de France\\
$^{\ddag}$ IRISA, Campus de Beaulieu, 35042 Rennes Cedex, France \\
}
}
\date{}

\begin{document}

\maketitle

\vspace{-0.5cm}
%========================================================================
\begin{abstract}
This paper is on homonymous  distributed 
systems where processes are prone to crash failures and have no initial 
knowledge of the  system membership (``homonymous'' means that several 
processes  may have the same identifier).  
New classes of failure detectors suited to these systems are first defined.
Among them,  the classes $\HO$ and $\HS$ are introduced 
that are the homonymous counterparts of the classes $\Omega$ and  $\Sigma$, 
respectively. (Recall that the pair $\langle \Omega,\Sigma\rangle$ 
defines the weakest failure detector to solve consensus.)
Then, the paper  shows how $\HO$ and $\HS$ can be implemented in homonymous 
systems without membership knowledge (under different synchrony 
requirements). Finally,  two algorithms are presented that use these 
failure detectors to solve consensus in homonymous asynchronous systems
where there is no initial knowledge of the membership. One algorithm 
solves consensus with $\langle \HO,\HS\rangle$, while the other 
uses only $\HO$, but needs a majority of correct processes. 

Observe that the systems with unique identifiers and anonymous 
systems are extreme cases of homonymous systems from which follows that
 all these results also apply to these systems.  Interestingly, the new
failure detector class $\HO$ can be implemented with partial synchrony, 
while the analogous class $\AO$ defined for anonymous systems 
can not be implemented (even in synchronous systems). Hence, the paper
provides us with the first proof showing that consensus can be solved in 
anonymous systems with only partial synchrony 
(and a majority of correct processes).
\end{abstract}
\begin{IEEEkeywords}
Agreement problem, Asynchrony, Consensus,  Distributed computability, 
failure detector, Homonymous system, 
Message-passing, Process crash.  
\end{IEEEkeywords}
%========================================================================

\section{Introduction}

\tightparagraph{Homonymous systems}
Distributed computing is on mastering uncertainty created by adversaries. 
The  first  adversary  is  of  course  the fact  that  the  processes are
geographically distributed which makes impossible to instantaneously obtain a
global state of the system. An  adversary can be static (e.g., synchrony 
or anonymity) or dynamic (e.g., asynchrony, mobility, etc.).  The  net 
effect of asynchrony and failures is the  most studied pair of adversaries.  
 
This paper is on agreement in  crash-prone  message-passing  distributed
systems. While  this topic  has been deeply  investigated in  the past
in the context of asynchrony and process failures 
(e.g., \cite{%AW04,
L96,R10}), we additionally consider here that 
several  processes can have  the  same identity,  i.e.,  the  
additional static adversary  that is {\it homonymy}. 
A motivation for homonymous processes in distributed 
systems can be found in \cite{DFGKRT11} where, for example, 
users keep their privacy taking their domain as their identifier 
(the same identifier is then
assigned to all the users of the same domain). Observe that 
homonymy is a generalization
of two cases: (1) having unique identifiers and 
(2) having the same identifier for all the processes (anonymity), 
which are the two extremes of homonymy.

We also assume  that the distributed system has to face
another static adversary, which is the fact that, initially, 
each process only 
knows its own identity. We say that the system has to work 
\emph{without initial knowledge of the membership}. 
This static adversary has been recently identified 
as of significant relevance in certain distributed contexts
\cite{DBLP:journals/ipl/JimenezAF06}.

\tightparagraph{How to face adversaries}
It is well-known that lots of problems cannot be solved in 
presence of some adversaries 
(e.g., \cite{DBLP:conf/stoc/Angluin80,DBLP:journals/jacm/AttiyaSW88,
DBLP:journals/jacm/FischerLP85,DBLP:journals/tpds/YamashitaK96}). 
%,DBLP:conf/podc/Zielinski08}). 
When  considering  process  crash  failures,  the  
{\it  failure  detector} approach introduced in 
\cite{DBLP:journals/jacm/ChandraHT96,DBLP:journals/jacm/ChandraT96} 
(see  \cite{R09}   for an 
introductory presentation) has  proved to be very attractive.  
It allows to enrich an otherwise too poor distributed  system 
to solve a given problem $P$,  in order to obtain a more powerful system 
in which $P$ can be solved. 

A failure detector is a  distributed oracle  that provides processes with
additional information related to  failed processes, and can consequently be
used to enrich the computability power of asynchronous  send/receive 
message-passing systems.
According to the  type (set  of process  identities, integers,
etc.) and  the quality  of  this  information,  several failure  detector
classes have been proposed. 
We refer the reader to  \cite{R10} where classes of 
failure detectors  suited to  agreement and communication
problems, corresponding failure detector-based algorithms, and 
additional behavioral assumptions that (when satisfied) allow 
these failure detectors to be implemented are presented.
It is interesting to observe that none of the original failure detectors
introduced in \cite{DBLP:journals/jacm/ChandraT96} can be implemented
without initial knowledge of the 
membership \cite{DBLP:journals/ipl/JimenezAF06}.

\tightparagraph{Aim of the paper}
Agreement problems are  central as soon as one wants  to capture the essence
of distributed computing. (If processes do not have to agree in one way
or another,  the problem we  have to solve  is not a  distributed computing
problem!)  
The  aim  of  this  paper   is  consequently  to  understand  the  type  of
information  on failures   that  is needed  when   one has    to solve   an
agreement  problem in  presence of asynchrony, process crashes,  
homonymy, and lack of initial knowledge of the membership. 
As consensus is the most central agreement problem we focus on it.

\tightparagraph{Related work}
As far as we know, consensus in anonymous networks has been addressed 
first in~\cite{DBLP:conf/wdag/BonnetR09,DFT09} 
(\cite{DFT09} considers different synchrony assumptions while 
\cite{DBLP:conf/wdag/BonnetR09} considers systems enriched with
failure detectors). Connectivity requirements for agreement 
in anonymous networks is addressed in~\cite{GT07}.

To the best of  our knowledge, up to now agreement  in  homonymous systems 
has been addressed  only  in \cite{DFGKRT11} and \cite{janus-opodis2011}.
In the former paper the  authors consider
that,  among the  $n$ processes,  up to  $t$ of  them can  commit Byzantine
failures.  The system is homonymous in the sense that there are $\ell$, 
$1 \leq \ell \leq n$,  different authenticated identities,  each process 
has one identity, and several processes can share the same identity. 
It is  shown in that paper that  $\ell >3t$  and  $\ell > \frac{3t+n}{2}$ 
are   necessary and sufficient conditions for solving consensus in synchronous 
systems and  partially synchronous systems, respectively.
The latter paper \cite{janus-opodis2011} mainly explores consensus
in a shared memory system with anonymous processes, and bounds
the complexity (namely, individual write and step complexities) of
solving consensus with the aid of an anonymous leader elector $\AO$ 
(see below).
They show that if the system is homonymous instead of purely anonymous
these bounds can be improved.

The consensus problem in anonymous asynchronous crash-prone message-passing 
systems  has been  recently addressed  in \cite{DBLP:conf/wdag/BonnetR09} 
(for the first time to our knowledge).  
In such systems, processes have no identity at 
all\footnote{They must also execute the same program, 
because otherwise they could use the program
(or a hash of it) as their identity. 
We consider that it is the same if processes have no identity 
or they have the same identity for all processes, 
since a process that lacks an identity can choose
a default value (e.g., $\bot$) as its identifier.}.
This paper introduces an  anonymous counterpart\footnote{In this paper,
when we say that a failure detector $A$ is the  {\it counterpart} 
of  a  failure   detector  $B$  we  mean  that,   in  a classical  
asynchronous system  (i.e., where each process  has its own identity) 
enriched    with  a  failure  detector of  class $A$,  it  is  possible  to
design  an  algorithm that builds  a failure detector  of the class  $B$
and vice-versa  by exchanging  $A$ and $B$.  Said differently, $A$  and $B$
have the same computability power in a classical crash-prone asynchronous 
system.}  (denoted $\NAP$ later in \cite{DBLP:conf/wdag/BonnetR10})  
of  the   perfect  failure detector  $P$
introduced in \cite{DBLP:journals/jacm/ChandraT96}. A failure detector
of class $\NAP$ returns an upper bound (that eventually becomes tight)
of the current number of alive processes. The paper then
shows that  there  is  an inherent  price  associated with anonymous
consensus,  namely,  while the lower bound on the number of rounds in 
a non-anonymous system enriched with $P$  is $t+1$ (where $t$ is 
the maximum number of faulty processes), it is $2t+1$  
in an anonymous system enriched with  $\NAP$.
The algorithm proposed assumes knowledge of the parameter $t$.

More general failure detectors suited to anonymous  distributed systems are
presented in \cite{DBLP:conf/wdag/BonnetR10}. 
Among other results,  this paper introduces the anonymous counterpart $A\Sigma$
of the quorum  failure detector  class 
$\Sigma$ \cite{DBLP:journals/jacm/Delporte-GalletFG10} and the 
anonymous counterpart $A\Omega$ of the eventual leader  failure detector 
class $\Omega$ \cite{DBLP:journals/jacm/ChandraHT96}. It also presents
the failure detector class $\AP$ which is the complement of $\NAP$.
An important result of \cite{DBLP:conf/wdag/BonnetR10} 
is the fact that  relations linking  failure detector classes 
are  not the same in non-anonymous systems   and anonymous systems. 
This is also the case
if processes do not know the number $n$ of processes in the system 
(unknown membership in anonymous systems). 
If $n$ is unknown, the equivalence between
$\mathit{AP}$ and $\overline{\mathit{AP}}$, 
shown in \cite{DBLP:conf/wdag/BonnetR10}, does not hold anymore.

Regarding implementability, it is stated 
in \cite{DBLP:conf/wdag/BonnetR10} that $\AO$ is not \emph{realistic}
(i.e., it can not be implemented in an anonymous  synchronous system 
\cite{DBLP:conf/dsn/Delporte-GalletFG02}). If the membership is
unknown, it is not hard to show that $\AP$ is not realistic either, 
applying similar techniques as those in
\cite{DBLP:journals/ipl/JimenezAF06}. On the other hand, 
while $\NAP$ can be implemented in an anonymous  
synchronous system, it is easy to show that it cannot be 
implemented in most partially synchronous systems 
(e.g., in particular, in those with all links eventually timely).

\tightparagraph{Contributions}
As mentioned, we explore the consensus problem in homonymous systems. 
Additional adversaries considered are
asynchrony, process crashes, and lack of initial knowledge of the membership. 
We can summarize the main contributions of this paper as follows.

First, the paper defines new classes  of failure detectors suited  
to homonymous
systems. These classes, denoted $\HO$ and  $\HS$, 
are shown to be homonymous counterparts of  $\Omega$ and  $\Sigma$, 
respectively.
The interest on the latter classes is motivated by the fact that  
$\langle \Sigma,\Omega\rangle$ is the weakest failure detector 
to solve consensus in crash prone asynchronous message-passing systems 
for any  number of  process failures  
\cite{DBLP:journals/jacm/Delporte-GalletFG10}.  
The paper also investigates the relations linking $\mathit{H\Sigma}$, 
$\mathit{A\Sigma}$ and  $\Sigma$, and shows that both 
$\HO$ and  $\HS$ can be obtained
from $\NAP$ in asynchronous anonymous systems. 
%As byproducts, we also introduce two new
%failure detector classes denoted  $\HP$ and $\MO$, which we consider of 
%independent interest.
%Class $\HP$ is the homonymous counterpart of  $\Diamond\overline{P}$
%(the complement of   $\Diamond P$  \cite{DBLP:journals/jacm/ChandraT96}).
%Class $\MO$, on its hand, is a generalization of $\AO$ in which, 
%instead of one, there may
%be many permanent leaders. It is shown that the new class $\MO$ is equivalent
%to $\HO$ in homonymous asynchronous systems.
As a byproduct, we also introduce a new
failure detector class denoted  $\HP$,
that is the homonymous counterpart of  $\Diamond\overline{P}$
(the complement of   $\Diamond P$  \cite{DBLP:journals/jacm/ChandraT96}),
which we consider of independent interest.

Then, the paper explores the implementability of these classes 
of failure detectors.
It presents an implementation of  $\HP$ in  homonymous message-passing  
systems  with  partially synchronous processes  and eventually timely  links.
This algorithm does  not require that the processes know the system
 membership. 
Since $\HO$ can be trivially implemented from $\HP$ without communication,
$\HO$ is realistic and can also be implemented in a partially synchronous 
homonymous system without membership knowledge. The paper also presents 
an implementation of  $\HS$
in  a  synchronous homonymous message-passing  system
without membership knowledge.

Finally, the paper presents two consensus algorithms for asynchronous 
homonymous systems enriched with $\HO$. Both algorithms are derived
from consensus algorithms for anonymous systems proposed in 
\cite{DBLP:conf/aina/BonnetR10} and
\cite{DBLP:conf/wdag/BonnetR10}, respectively. 
The main challenge, and hence, the main contribution
of our algorithms, is to modify the original algorithms that 
used $\AO$ to use $\HO$ instead.
In the second algorithm, also the use of $\AS$ has been 
replaced by the use of $\HS$.

The first algorithm assumes that each process knows the value $n$
and that a majority of processes is correct in 
all executions\footnote{The knowledge of $n$ can be replaced 
by the knowledge of a
parameter $\alpha$ such that, $\alpha>n/2$ and, in all executions, 
at least $\alpha$ processes are correct.}. Since, as mentioned,
$\HO$ can be implemented with partial synchrony, the combination of the
algorithms presented (to implement $\HO$ and to solve consensus with $\HO$)
form a distributed algorithm that solves consensus in any homonymous system
with partially synchronous processes, eventually timely  links, and 
a majority of correct processes. Applied to anonymous systems,
this result relaxes the known conditions to solve consensus, since previous
algorithms were based on unrealistic failure 
detectors ($\AO$) or failure  detectors that require
a larger degree of synchrony ($\NAP$).

The second consensus algorithm presented works for any number of process 
crashes, and does not need to know $n$, but assumes that 
the system is  enriched with the pair of failure detectors 
$\langle \HS, \HO\rangle$.
This algorithm, combined with the algorithms to implement 
$\HS$ and $\HO$, shows that the consensus problem 
can be solved in \emph{synchronous} homonymous systems subject to
any number of crash failures without the initial knowledge neither of 
the parameter $t$  nor  of the membership.
Applied to anonymous systems, this result relaxes the known conditions to 
solve consensus under any number of failures, 
since previous algorithms used unrealistic detectors
($\AO$) or required to know $t$ or an upper bound on it.

%This second consensus algorithm also forces us to restate 
%the conjecture of which could be the weakest failure detector to solve
%consensus in asynchronous anonymous systems. The algorithm solves 
%consensus in anonymous systems with a pair of detectors 
%$\langle \HS, \HO\rangle$, but since classes $\MO$ and $\HO$ are
%equivalent, it can be modified to solve consensus 
%with a pair $\langle \HS, \MO\rangle$. As mentioned,
%it is shown here that $\HS$ can be obtained from $\AS$, and both
%$\HS$ and $\HO$ can be obtained from $\NAP$. Additionally,
%any failure detector in $\AO$ is also in $\MO$.
%The conjecture issued in \cite{DBLP:conf/wdag/BonnetR10} was that
%$\langle \AS, \AO\rangle \oplus \NAP$ \footnote{$\oplus$ represents 
%a form of composition in which the resulting failure detector 
%outputs $\bot$ for a finite time until it behaves at all processes 
%as  one -and the same- of the two detectors  that are combined.} 
%could be the weakest failure detector. 
%Then, the new candidate  to be the 
% weakest failure detector for consensus despite anonymity is  now 
%$\langle \HS, \MO\rangle$ (or, equivalently, $\langle \HS, \HO\rangle$).

This second consensus algorithms also forces us to restate 
the conjecture of which could be the weakest failure detector to solve
consensus in asynchronous anonymous systems. The algorithm solves 
consensus in anonymous systems with a pair of detectors 
$\langle \HS, \HO\rangle$,
and we describe how it can be modified to solve consensus 
with a pair $\langle \HS, \AO\rangle$. Additionally, as mentioned,
it is shown here that $\HS$ can be obtained from $\AS$, and both
$\HS$ and $\HO$ can be obtained from $\NAP$. 
The conjecture issued in \cite{DBLP:conf/wdag/BonnetR10} was that
$\langle \AS, \AO\rangle \oplus \NAP$ \footnote{$\oplus$ represents 
a form of composition in which the resulting failure detector 
outputs $\bot$ for a finite time until it behaves at all processes 
as  one -and the same- of the two detectors  that are combined.} 
could be the weakest failure detector. 
Then, using the same algorithm described in \cite{DBLP:conf/wdag/BonnetR10} 
to combine the
consensus algorithms for $\langle \HS, \AO\rangle$ and  
$\langle \HS, \HO\rangle$, the new candidate  to be the 
 weakest failure detector for consensus despite anonymity is  now 
$\langle \HS, \AO\rangle \oplus \langle \HS, \HO\rangle$.

\tightparagraph{Roadmap}
The   paper  is   made   up  of   \ref{sec:FDs-based-consensus}  sections.   
Section \ref{sec:model} presents the system model. Section 
\ref{sec:FDs-defintions} introduces failure  detector classes suited 
to homonymous
systems, and explores their relation with other classes and their 
implementability.
%while   Section \ref{sec:FDs-implementation} focuses on their implementation. 
Finally, Section \ref{sec:FDs-based-consensus} presents  failure 
detector-based homonymous consensus algorithms.
%Finally, Section \ref{sec:conclusion} concludes the paper. 

%========================================================================

\section{System Model}
\label{sec:model} 

\tightparagraph{Homonymous processes} 
Let $\Pi$ denote the set of processes with  $|\Pi|=n$.
 We use  $id(p)$ to denote the identity  of process  $p \in \Pi$.  
Different processes may  have the
same  identity,  i.e. $p  \neq  q  \nRightarrow  id(p) \neq  id(q)$.  Two
processes with the same identity are said to be \emph{homonymous.}   
Let $S \subseteq  \Pi$ be any subset of processes. We  define $I(S)$ as the
\emph{multiset} (sometimes also called {\it bag}) of process identities  
in $S$,  $I(S) =  \{id(p) :  p \in S\}$. Let us remember that, differently
from a set, an element of a multiset can appear more than once. 
Hence, as  $I(S)$ may contain several  times the same   identity, 
we always have $|I(S)|=|S|$.  
The  \emph{multiplicity}   (number  of   instances)  of
identity    $i$   in    a    multiset   $I$ is  denoted   
$\mathit{mult}_I(i)$. When $I$ is clear from the context we will use simply
$\mathit{mult}(i)$.     $P(I) \subseteq  \Pi$ is used to denote
the  processes whose  identity is in the multiset $I$, i.e., 
$P(I)=\{p: p  \in \Pi \wedge id(p) \in I\}$. 
Every  process  $p \in \Pi$  knows  its   own identity
$id(p)$. Unless  otherwise stated, a process  $p$ does not  know the system
membership $I(\Pi)$, nor the system size $n$, nor any upper bound $t$ on the 
number of faulty processes. Observe that the set $\Pi$ is
a formalization tool that is not known by the set of processes of the system.

Processes are  asynchronous, unless otherwise  stated. We assume  that time
advances at discrete  steps. We assume a global clock  whose values are the
positive natural numbers, but processes cannot access it. 
Processes can  fail by  crashing, i.e.,   stop taking steps. 
A process that  crashes in a run is said to  be \emph{faulty} and a
process that is not  faulty in a run is said to  be \emph{correct}. 
The set of correct processes is denoted by $\C \subseteq \Pi$.

\tightparagraph{Communication}
The processes can invoke the primitive $broadcast(m)$ to send a message
$m$ to all  processes of the system (including  itself). This communication
primitive is modeled in the following way. The network is assumed to have a
directed link  from process $p$ to  process $q$ for each  pair of processes
$p,q  \in  \Pi$  ($p$ does  not  need  to  be  different from  $q$).  Then,
$broadcast(m)$ invoked at  process $p$ sends one copy  of message $m$ along
the link  from $p$ to $q$, for  each $q \in \Pi$.  Unless otherwise stated,
links are asynchronous and reliable,  i.e., links neither lose messages nor
duplicate messages nor corrupt  messages nor generate spurious messages. If
a process crashes while broadcasting a message, the message is received
by an arbitrary  subset of processes.

\tightparagraph{Notation and  time-related definitions}
The  previous  model is  denoted $\mathit{HAS}[\emptyset]$  (Homonymous
Asynchronous System). 
We  use  $\mathit{HPS}[\emptyset]$  to  denote a  homonymous  system  where
processes  are partially  synchronous and  links are  eventually  timely. 
A
process is  \emph{partially synchronous} if the  time to execute  a step is
bounded, but  the bound is unknown.  A link is  \emph{eventually timely} if
there  is  an unknown  global  stabilization  time  (denoted $GST$)
 after  which  all
messages sent  across the  link are delivered  in a bounded  $\delta$ time,
where  $\delta$ is  unknown.  Messages sent  before  $GST$ can  be lost  or
delivered after an arbitrary (but finite) time. 

$AS[\emptyset]$ denotes the  classical asynchronous  system with
unique identities  and  reliable channels. Finally, $AAS[\emptyset]$ denotes
the  Anonymous  Asynchronous   System model  \cite{DBLP:conf/wdag/BonnetR10}.  
Observe  that  $AS[\emptyset]$  and $AAS[\emptyset]$  are  special  cases
(actually extreme cases with respect to homonymy) of 
$\mathit{HAS}[\emptyset]$ (an anonymous system can be seen 
as a homonymous system where
all  processes  have the  same default identifier $\bot$).

%========================================================================
\section{Failure Detectors}
\label{sec:FDs-defintions}

In this section we define failure detectors previously proposed and the 
ones proposed here for homonymous systems. Then, relationships between 
these detectors are derived, and their implementability is explored.

\tightparagraph{Failure detectors for classical and anonymous systems}
We briefly describe here some failure detector previously proposed. 
%Detailed definitions have been included in Appendix~\ref{subsection:FD}.
We start with the classes that have been defined for $AS[\emptyset]$.

%\begin{itemize}
%\item
\emph{A failure detector of class $\Sigma$} 
\cite{DBLP:journals/jacm/Delporte-GalletFG10} 
provides each process  
$p \in \Pi$ with a variable $trusted_p$ which contains a set of process 
identifiers. The properties that are satisfied by these sets are 
[Liveness] $\forall p \in \C , \exists \tau \in N: 
\forall \tau' \ge \tau , $ $trusted_p^{\tau'} \subseteq I(\C)$, and
[Safety] $\forall p, q \in \Pi, \forall \tau, \tau' \in N,
 trusted_p^{\tau} \cap trusted_q^{\tau'} \not=\emptyset$.
 
%\item
\emph{A failure detector of class $\Omega$}
 \cite{DBLP:journals/jacm/ChandraHT96}
provides each process  $p \in \Pi$ with a variable $leader_p$ such that 
[Election] eventually all these variables contain the same process 
identifier of a correct process.
%\end{itemize}

The following failure detector classes have been defined for
 anonymous systems $AAS[\emptyset]$.
 
%\begin{itemize}
%\item
\emph{A failure detector of class $\AO$} \cite{DBLP:conf/wdag/BonnetR10}
provides each process $p \in \Pi$ with a variable $a\_leader_p$, such 
that [Election] there is a time after which, permanently, (1) there is 
a correct process whose Boolean variable is true,
 and (2) the Boolean variables of the other correct processes are false.
 
%\item
\emph{A failure detector of class $\NAP$} \cite{DBLP:conf/wdag/BonnetR09}
provides each process  $p \in \Pi$ with a variable $anap_p$ such that, 
if $anap_p^{\tau}$ and $\C^{\tau}$ denote the value of this variable and 
the number of alive processes at time $\tau$, respectively, then [Safety]
 $\forall p \in \Pi, \forall \tau \in N, anap_p^{\tau} \ge |\C^{\tau}|$, 
and [Liveness]
 $\exists  \tau \in N, \forall p \in \C, 
\forall \tau' \geq \tau, anap_p^{\tau'}=|\C|$.

%\item
\emph{A failure detector of class $\AS$} \cite{DBLP:conf/wdag/BonnetR10}
provides each process  $p \in \Pi$ with a variable $a\_sigma_p$ that 
contains a set of pairs of the form $(x,y)$. The parameter $x$ is a 
label provided by the failure detector, and $y$ is an integer. 
Let us denote $a\_sigma_p^\tau$ the value of variable $a\_sigma_p$ at 
time $\tau$. 
Let $S_A(x)=\{p \in \Pi \ |\ \exists \tau \in N : (x,-) \in a\_sigma_p^\tau\}$.
Any failure detector of class $\AS$ must satisfy the following properties: 

\begin{itemize}
\item
Validity.  No set $a\_sigma_p$ ever contains simultaneously two pairs 
with the same label. 
\item Monotonicity. 
$ \forall p \in \Pi , \forall \tau \in N : 
(((x,y) \in a\_sigma_p^{\tau}) \implies 
(\forall \tau' \ge \tau: \exists y' \leq y : (x,y') \in a\_sigma_p^{\tau'})$.
\item
Liveness.
 $\forall p \in \C , \exists \tau \in N: 
  \forall \tau' \ge \tau : \exists (x,y) \in a\_sigma_p^{\tau'} : 
   (|S_A(x) \cap  \C| \ge y)$.
\item
Safety.
$\forall p_1, p_2 \in \Pi , \forall \tau_1, \tau_2 \in N,  
\forall (x_1,y_1) \in a\_sigma_{p_1}^{\tau_1} : 
\forall (x_2,y_2) \in a\_sigma_{p_2}^{\tau_2}: 
 \forall T_1 \subseteq S_A(x_1): \forall T_2 \subseteq S_A(x_2) : 
((|T_1| = y_1) \wedge (|T_2| = y_2)) \implies (T_1 \cap T_2 \not= \emptyset)$.
\end{itemize}
%\end{itemize}

\tightparagraph{Failure detectors for homonymous systems}
Classical failures detectors output a set of processes' identifiers. 
Our failures detectors extend this output to 
a multiset of processes' identifiers, due to the homonymy nature of the system.
%That is, several elements of the multiset may have the same value 
%(i.e, identifier). 
%Note that in an homonymous system, to know an identifier does not
%implies necessarily to know the process itself.
The following are the new failure detectors proposed for homonymous systems. 
%(More details in Appendix~\ref{section:FDH}.)

%\begin{itemize}
%\item
\emph{A failure detector of class $\HP$} eventually 
outputs forever the multiset with the 
identifiers of the correct processes. More formally,
a failure detector of class $\HP$ provides each process  $p \in \Pi$ 
with a variable $\trusted_p$, such that [Liveness]
$\forall p \in \C, \exists \tau \in N  : \forall \tau' \ge \tau$, 
$\trusted_p^{\tau'} = I(\C)$. This failure detector $\HP$ is the 
counterpart of $\Diamond\overline{P}$.

%\item
\emph{A failure detector of class $\HO$} 
eventually outputs the same identifier $\ell$
and number $c$ at all processes, such that $\ell$ is the identifier of 
some correct process, and $c$ is the number of correct processes that 
have this identifier $\ell$. More formally,
a failure detector of class $\HO$ provides each process  $p \in \Pi$ 
with two variables $\leader_p$ and $\multiplicity_p$, such that 
[Election] $\exists \ell \in I(\C), \exists \tau \in N  : 
\forall \tau' \ge \tau, \forall p \in \C$, $\leader_p^{\tau'}=\ell$, 
and $\multiplicity_p^{\tau'}=\mathit{mult}_{I(\C)}(\ell)$.

Any correct process $p$ such that $id(p)=\ell$ is called a \emph{leader}.
Note that this failure detector does not choose only one leader, like 
in $\Omega$ or in $\AO$, but a set of leaders with the same identifier. 
When all identifiers are different, the class $\HO$ is equivalent to $\OO$. 
Furthermore, a failure detector of class $\HO$ can be obtained from any 
detector $D$ of class $\HP$ without any communication (for instance, 
setting at each process $p$ periodically $\leader_p$ to the smallest 
element in $D.\trusted_p$, and 
$\multiplicity_p \leftarrow \mathit{mult}_{D.\trusted_p}(\leader_p)$).

%%\item
%\emph{A failure detector of class $\MO$} 
%satisfies that eventually a non-empty set 
%$L$ of correct processes will permanently identify themselves as leaders,
%and all the processes will know the size of this set $|L|$.
%More formally, a failure detector of the class $\MO$ provides each process 
%$p \in \Pi$ with two variables, a boolean variable $\mleader_p$ and an 
%integer variable $\cardinality_p$, such that [Election] 
%$\exists L \subseteq \C, L \neq \emptyset, \exists \tau \in N  : 
%\forall \tau' \ge \tau, (\mleader_p^{\tau'}=\mathrm{true}) \iff (p \in L)$, 
%and $\cardinality_p^{\tau'}=|L|$.

%\item
\emph{A failure detector of class $\HS$} 
provides each process $p \in \Pi$ 
with two variables $\quora_p$ and $\labels_p$, where $\quora_p$ is a set 
of pairs of the form $(x,m)$ ($x$ is a label, and $m$ is a multiset such 
that $m \subseteq I(\Pi)$) and $\labels_p$ is a set of labels. 
Roughly speaking, each pair $(x,m)$ determines a set of quora, and the set 
$\labels_p$ of a process $p$ determines in which of these sets it 
participates. More formal,
let us denote $\quora_p^{\tau}$ and  $\labels_p^{\tau}$ the values of 
variables $\quora_p$ and $\labels_p$ at time $\tau$, respectively.
Let $S(x)=\{p \in \Pi \ |\ \exists \tau \in N : x \in \labels_p^\tau\}$.
Any failure detector of class $\HS$ must satisfy the following properties: 
\begin{itemize}
\item
Validity.  No set $\quora_p$ ever contains simultaneously two pairs 
with the same label.
\item
Monotonicity.
$ \forall p \in \Pi , \forall \tau \in N, \forall \tau' \ge \tau$: 
(1) $\labels_p^{\tau} \subseteq \labels_p^{\tau'}$, and 
(2) $((x,m) \in \quora_p^\tau) \implies  
\exists m' \subseteq m : (x,m') \in \quora_p^{\tau'} $.
\item
Liveness. $\forall p \in \C , \exists \tau \in N: \forall \tau' \ge \tau , 
\exists (x,m) \in \quora_p^{\tau'} : m \subseteq I(S(x) \cap  \C)$.
\item
Safety. $\forall p_1, p_2 \in \Pi , \forall \tau_1, \tau_2 \in N,  
\forall (x_1,m_1) \in \quora_{p_1}^{\tau_1} : 
\forall (x_2,m_2) \in \quora_{p_2}^{\tau_2}:
 \forall Q_1 \subseteq S(x_1), \forall Q_2 \subseteq S(x_2), 
(I(Q_1) = m_1 \wedge I(Q_2) = m_2) \implies (Q_1 \cap Q_2 \not= \emptyset)$.
\end{itemize}
Comparing $\HS$ and $\AS$, one can observe that $\HS$ has pairs $(x,m)$ 
in which $m$ is a multiset of identifiers, while $\AS$ uses pairs $(x,y)$ 
in which $y$ is an integer. However,
a more important difference is that, in $\HS$, each process has two variables.
 Then, the labels that a process $p$ has in $\quora_p$ can be disconnected 
from those it has in
$\labels_p$. This allows for additional flexibility in $\HS$.
%\end{itemize}

\tightparagraph{Reductions between failure detectors}
\label{red-FD}
In this section we claim that it can be shown, via reductions, 
the relation of the newly defined failure detector classes with 
the previously defined classes.
We use the standard form of comparing the relative power of failure 
detector classes of \cite{DBLP:journals/jacm/ChandraT96}.
A failure detector class $X$ is \emph{stronger} than class $X'$ in 
system $Y[\emptyset]$ if there is an algorithm $A$ that emulates the 
output of a failure detector of class $X'$ in $Y[X]$ 
(i.e., system $Y[\emptyset]$ enhanced with a failure detector $D$ 
of class $X$). We also say that $X'$ can be obtained from $X$ 
in $Y[\emptyset]$. Two classes are equivalent if this property 
can be shown in both directions.

%Due to space restrictions, we 
We only present here the main results. 
The proofs and additional details can be found in the Appendix.
The first result shows that, 
in classical systems with unique identifiers, $\Sigma$, $\HS$, 
and $\AS$ are equivalent.

\begin{theorem}
\label{thm:equivalent}
Failure detector classes $\Sigma$, $\HS$, and $\AS$ are equivalent 
in $AS[\emptyset]$. Furthermore,
the transformations between $\Sigma$ and $\HS$ do not require 
initial knowledge of the membership.
\end{theorem}

%The following result shows the equivalence between the two new 
%failure detector classes $\HO$ and $\MO$ in homonymous asynchronous systems.
%\begin{theorem}
%\label{thm:hotomo}
%Class $\MO$ can be obtained from class $\HO$ in $\mathit{HAS}[\emptyset]$
% without communication. Conversely, class $\HO$ can be obtained 
%from class $\MO$ in $\mathit{HAS}[\emptyset]$.
%\end{theorem}
%
%It can be trivially observed that any detector of class $\AO$ also belongs
%to class $\MO$, where it is fixed $\cardinality_p=1$ at every process $p$.

In anonymous systems we have the following properties. Recall that 
an anonymous system is assumed to be a homonymous system 
in which every process has a default identifier $\bot$\footnote{Note that this
differs from the assumption used in \cite{DBLP:conf/wdag/BonnetR10}.}.

\begin{theorem}
\label{thm:astohs}
Class $\HS$ can be obtained from class $\AS$ in 
$AAS[\emptyset]$ without communication.
\end{theorem}

\begin{theorem}
\label{thm:nap}
Classes $\HP$ and $\HS$ can be obtained from class $\NAP$ 
in $AAS[\emptyset]$ without communication.
\end{theorem}

\begin{figure}[htbp]
\begin{center}
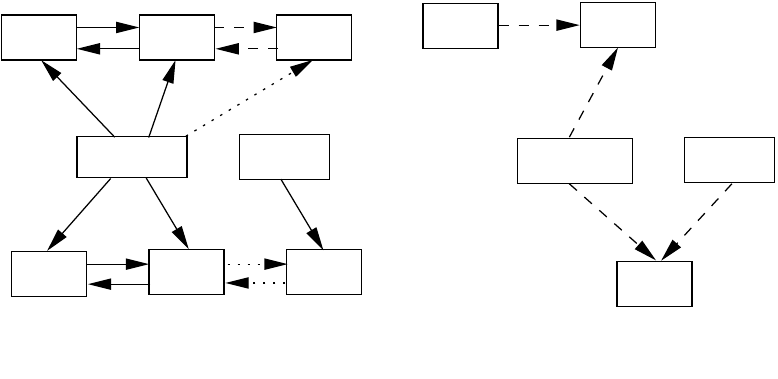
\caption{Relations   between  failure  detector   classes  in   the  models
$AS[\emptyset]$  and $AAS[\emptyset]$. 
There  is  an  arrow from  class  $X$  to  $X'$  if  $X$ is  stronger  that
$X'$. Solid arrows are 
relations       shown       by       Bonnet       and       Raynal       in
\cite{DBLP:conf/wdag/BonnetR10}. Dashed arrows are 
relations shown here, while dotted arrows are trivial relations.}
\label{default}
\end{center}
\end{figure}

%========================================================================

\section{Implementing Failure Detectors in  Homonymous Systems}
\label{appendix:FDs-implementation}

%\tightparagraph{Implementing failure detectors in  homonymous systems}
%\label{sec:FDs-implementation}
%
In this section, we show that there are algorithms that implement 
the failure detectors classes $\HP$ and $\HO$ in $HPS[\emptyset]$
 (homonymous partially synchronous system). We also implement the 
failure detector $\HS$ in $HSS[\emptyset]$ (homonymous synchronous system).
 In all cases they do not need to know initially the membership. 
%Due to the limits of space, we present the main results, 
%leaving the details in Appendix~\ref{appendix:FDs-implementation}.

%\begin{theorem}
%There is an algorithm that implements a failure detector of the class ${\HP}$ in a system $\mathit{HPS}[\emptyset]$, even if the membership is not known initially.
%\end{theorem}
%As mentioned above, a failure detector of the class ${\HP}$ can be used to obtain a failure detector of class $\HO$ without communication. 
%\begin{theorem}
%There is an algorithm that implements a failure detector of class ${\HS}$ in a system $HSS[\emptyset]$,  even if the membership is not known initially. 
%\end{theorem}

%============================================================

\begin{figure}
\begin{lstlisting}
Init
	$h\_trusted_p$ := $\emptyset$;  // multiset of process identifiers
	$mship_p$ := $\emptyset$; // set of process identifiers
	$r_p$ := 1;
	$timeout_p$ := 1;
	start Tasks T1 and T2;

Task T1
	repeat forever (' \label{f1-9}')
		broadcast $(\mathit{POLLING}, r_p, id(p))$;
		wait $timeout_p$ time;
		$tmp_p$ := $\emptyset$; // $tmp_p$ is an auxiliary multiset
		for each $(\mathit{P\_REPLY}, r, r', id(p), id(q))$ received 
		             with $(r \le r_p \le r')$ do (' \label{f1-13}')
			add one instance of $id(q)$ to $tmp_p$;
		end for; (' \label{f1-15} ')
		$h\_trusted_p$ := $tmp_p$; (' \label{f1-16}')
		$r_p$ := $r_p + 1$; (' \label{f1-17}')
	end repeat; (' \label{f1-18}')

Task T2	
	upon reception of $(\mathit{POLLING}, r_q, id(q))$ do (' \label{f1-21} ')
		if $id(q) \notin mship_p$ then (' \label{f1-22} ')
			$mship_p$ := $mship_p \cup \{id(q)\}$;
			create $latest\_r_p[id(q)]$;
			$latest\_r_p[id(q)]$ := 0;
		end if; (' \label{f1-26} ')
		if $latest\_r_p[id(q)] < r_q$ then (' \label{f1-27}')
		  broadcast   $(\mathit{P\_REPLY},  latest\_r_p[id(q)]  + 1,  r_q,  id(q),id(p))$; (' \label{f1-28} ')
		end if;  (' \label{f1-29} ')
		$latest\_r_p[id(q)]$ := $\max(latest\_r_p[id(q)], r_q)$;  (' \label{f1-30} ')
	
	upon reception of $(\mathit{P\_REPLY}, r, r', id(p), -)$ with ($r < r_p$) do (' \label{f1-32} ')
		$timeout_p$ := $timeout_p + 1$; (' \label{f1-33} ')

\end{lstlisting}		
\caption{Algorithm that implements $\HP$ (code for process $p$).}
\label{Fig-HP}
\end{figure}

%-----------------------------------------------------------

\subsection{Implementation of $\HP$ and $\HO$}
The algorithm of Figure \ref{Fig-HP} implements ${\HP}$ (and ${\HO}$ 
with trivial changes) in $\mathit{HPS}[\emptyset]$ where processes 
are partially synchronous, links are eventually timely, and membership 
is not known. 

\tightparagraph{Brief description of the algorithm:} 
It is a polling-based algorithm that executes in rounds. 
At every round $r$, the Task 1 of each process $p$ broadcasts 
$(POLLING,r,id(p))$ messages. After a time $timeout_p$, 
it gathers in the variable $tmp_p$ (and, hence, also in $\trusted_p$) 
a multiset with the senders' identifiers $id_s$ of processes from 
$(P\_REPLY,r',r'',id(p),id_s)$ messages received with $r' \le r \le r''$.

Task 2 is related with the reception of $POLLING$ and $P\_REPLY$ 
messages. When a process $p$ receives a $(POLLING,r,id(q))$ message from 
process $q$, process $p$ has to respond with as many $P\_REPLY$ 
as process $q$ needs to receive up to round $r$, and not 
previously sent by process $p$ (Lines~\ref{f1-27}-\ref{f1-29}). 
Note that the $P\_REPLY$ 
messages are piggybacked in only one message (Line~\ref{f1-28}). 
Also note that is in variable $latest\_r_p[id(q)]$ where $p$ holds the 
latest round broadcast to $id(q)$. 
If it is the first time that process $p$ receives a $(POLLING,-,id)$ 
message from a process with identifier 
$id$, then variable $latest\_r_p[id]$ is created and initialized to 
zero (Lines~\ref{f1-22}-\ref{f1-26}).

It is important to remark that, for each different identifier $id$, 
only one $(P\_REPLY,-,-,id(q),id)$ 
message is broadcast by each process $q$. So, if processes $v$ and 
$w$ with $id(v)=id(w)=x$ broadcast two $(POLLING,r,x)$ messages, 
then each process $p$ only broadcast one $(P\_REPLY,r',r'',x,q)$ 
message with  $r' \le r \le r''$. Note that eventually (at least after 
GST time) each $P\_REPLY$  
message sent by any process has to be received by all correct processes. 
Hence, eventually processes $v$ and $w$ 
will receive all $P\_REPLY$  messages generated due to $POLLING$ messages.

Finally, Lines~\ref{f1-32}-\ref{f1-33} of Task 2 allow process $p$ 
to adapt the variable 
$timeout_p$ to the communication latency and process speed. 
When process $p$ receives an outdated $(P\_REPLY,r,-,id(p),-)$ 
message (i.e., a message with round $r$ less than current round $r_p$), 
then it increases its variable $timeout_p$.

\begin{lemma}
\label{l-faulty}
Given processes $p \in \C$ and $q \notin \C$, there is a round $r$ 
such that $p$ does not receive any 
$(\mathit{P\_REPLY}, \rho, \rho', id(p) , id(q))$ message from 
$q$ with $\rho' \ge r$.
\end{lemma}
\begin{proof}
There is a time $\tau$ at which $q$ stops taking steps. If $q$ ever 
sent a $(\mathit{P\_REPLY}, -, - , id(p) , id(q))$ message, 
consider the largest $x$
such that $q$ sent message $(\mathit{P\_REPLY}, -, x, id(p) , id(q))$. 
Otherwise, let $x=0$. Then, the claim holds for $r=x+1$.
\end{proof}

\begin{lemma}
\label{l-correct}
Given processes $p,q \in \C$, there is a round $r$ such that, for all 
rounds $r' \geq r$, 
when $p$ executes the loop of Lines~\ref{f1-13}-\ref{f1-15} 
with $r_p = r'$, it has received 
a message 
$(\mathit{P\_REPLY}, \rho, \rho', id(p) , id(q))$ from $q$ with 
$\rho \le r' \le \rho'$.
\end{lemma}
\begin{proof}
Observe that, since $p$ is correct, it will repeat forever the loop of 
Lines~\ref{f1-9}-\ref{f1-18}, with the value of $r_p$ 
increasing in one unit at each iteration.
Hence, $p$ will be sending forever messages $(\mathit{POLLING}, -, id(p))$ 
after $GST$ with increasing round numbers, that will eventually
 be received by $q$. Then, $q$ eventually will send infinite 
$(\mathit{P\_REPLY}, -, - , id(p) , id(q))$ messages after $GST$, 
with increasing round numbers. 
Let $(\mathit{P\_REPLY}, x, - , id(p) , id(q))$ be the first such 
message sent by $q$ after $GST$. Then, for each round number $y \ge x$, 
there is some message $(\mathit{P\_REPLY}, \rho, \rho', id(p) , id(q))$ 
sent by $q$ with $\rho \le y \le \rho'$, and these messages are delivered at
 $p$ at most $\delta$ time after being sent.

Now, assume for contradiction that for each round $y \geq x$, 
there is a round $y' \geq y$ such that, when $p$ executes the loop 
of Lines~\ref{f1-13}-\ref{f1-15} with $r_p = y'$, 
it has not received the message 
$(\mathit{P\_REPLY}, \rho, \rho', id(p) , id(q))$ from $q$ with 
$\rho \le y' \le \rho'$. But, every time this happens, when the message 
is finally received,
$r_p$ has been incremented in Line~\ref{f1-17} and, hence, 
$\mathit{timeout}_p$ 
is incremented (in Lines~\ref{f1-32}-\ref{f1-33}). 
Then, eventually, by some round $r$, 
the value of $\mathit{timeout}_p$ will be greater than $2 \delta + \gamma$, 
where $\gamma$ is the maximum time that $q$ takes 
to execute Lines~\ref{f1-21}-\ref{f1-30}. 
Then, $p$ will receive 
message $(\mathit{P\_REPLY}, \rho, \rho', id(p) , id(q))$ 
with $\rho \le r' \le \rho'$ before executing the loop of 
Lines~\ref{f1-13}-\ref{f1-15} 
with $r_p = r'$, for all $r' \geq r$. We have reached a contradiction 
and the claim of the lemma follows.
\end{proof}

\begin{theorem}
The algorithm of Figure \ref{Fig-HP} implements a failure detector of 
the class ${\HP}$ in a system $\mathit{HPS}[\emptyset]$ (homonymous 
system where processes are partially synchronous and links are eventually 
timely), even if the membership is not known initially.
\end{theorem}
\begin{proof}
Consider a correct process $p$. From Lemma~\ref{l-faulty}, there is a 
round $r$ such that $p$ does not receive any 
$(\mathit{P\_REPLY}, \rho, \rho', - , -)$ message with
 $\rho' \ge r$ from any faulty process. From Lemma~\ref{l-correct}, 
there is a round $r'$ such that for all rounds $r'' \geq r'$, when 
$p$ executes the loop of Lines~\ref{f1-13}-\ref{f1-15} 
with $r_p = r''$, it has 
received a $(\mathit{P\_REPLY}, \rho, \rho', -, -)$ message with 
$\rho \le r'' \le \rho'$
from each correct process. Hence, for every round $r'' \geq \max(r,r')$ 
when the Line~\ref{f1-16} is executed with $r_p=r''$, 
the variable $h\_trusted_p$ 
is updated with the multiset $I(\C)$.
\end{proof}

We can obtain $\HO$ from the algorithm of Fig.~\ref{Fig-HP} without 
additional communication. This can be done by simply including, 
immediately after Line~\ref{f1-16}, $\leader_p \leftarrow \min(\trusted_p)$  
(i.e., the smallest identifier in $\trusted_p$) and $\multiplicity_p 
\leftarrow \mathit{mult}_{\trusted_p}(\leader_p)$.
 
\begin{corollary}
The algorithm of Figure \ref{Fig-HP} can be changed to implement a 
failure detector of the class $\HO$ in a system $\mathit{HPS}[\emptyset]$ 
(homonymous system where processes are partially synchronous and links 
are eventually timely), even if the membership is not known initially.
\end{corollary}

%============================================================

\subsection{Implementation of ${\HS}$}
Figure \ref{Fig-HS} implements ${\HS}$ in $HSS[\emptyset]$] 
where processes are synchronous, links are timely, and
membership is not known. 

\tightparagraph{Brief description of the algorithm} 
It runs in synchronous steps. In each step every process $p$ 
broadcasts a $(IDENT,id(p))$ message. Then, process $p$ waits for
$(IDENT,-)$ messages sent through reliable links 
in this synchronous step by alive processes. Process $p$ 
gathers in the multiset variable $mset_p$ the 
identifiers $id$ of all $(IDENT,id)$ messages received. 
At the end of this step, variables $\quora_p$ and $\labels_p$ 
are updated with the value of $mset_p$. 
Note that for process $p$ the label $x$ of a quorum $(x,m)$ is formed by the 
multiset $mset_p$ (i.e, $x=m=mset_p$).

\begin{theorem}
The algorithm of Figure \ref{Fig-HS} implements a failure detector 
of the class ${\HS}$ in a system $\mathit{HSS}[\emptyset]$ (homonymous 
synchronous systems), even if the membership is not known initially.
\end{theorem}
\begin{proof}
From the definition of $\HS$, it is enough to prove the following properties.

%\begin{itemize}
%\item 
\emph{Validity.}
Since $\quora_p$ is a set, and the elements included in it are of the 
form $(mset,mset)$ (see Line~\ref{f2-7} in Figure \ref{Fig-HS}) there cannot 
be two pairs with the same label.
   
%\item 
\emph{Monotonicity.}
The monotonicity of $\labels_p$ in Figure \ref{Fig-HS} holds because 
$\labels_p$ is initially empty, and 
each step, $\labels_p$ either grows or remains the same 
(see Line~\ref{f2-8} in Figure \ref{Fig-HS}). 
Similarly, the monotonicity of $\quora_p$ in Figure \ref{Fig-HS} 
follows from the fact that $\quora_p$ is 
initially empty, and any element $(mset,mset)$ included in it is 
never removed (see Line~\ref{f2-7} in Figure \ref{Fig-HS}).

%\item 
\emph{Liveness.}
Let $s$ be the synchronous step in which the last faulty process crashed. 
Then, in every step $s'$ after $s$ only correct processes will execute. 
Consider any process $p \in \C$. In step $s'$ will receive messages 
from all correct 
processes, and, hence, $mset_p=I(\C)$. 
Then, process $p$ includes $(I(\C),I(\C))$ in $\quora_p$, and $I(\C)$ in  
$\labels_p$. Therefore, each correct process $p$ is in $S(I(\C))$. 
So, after step $s$, for each correct process $p$, the pair $(I(\C),I(\C))$ 
is in $\quora_p$, and $I(\C)=I(S(I(\C)) \cap  \C)$. 
 
%\item 
\emph{Safety.}
Consider two pairs $(x_1,x_1) \in \quora_{p_1}^{\tau_1}$ and 
$(x_2,x_2) \in \quora_{p_2}^{\tau_2}$, 
for any $p_1, p_2 \in \Pi$ and any $\tau_1, \tau_2 \in N$. 

Let $M_1$ be the set of processes from which $p_1$ received 
$(\mathit{IDENT},-)$ messages in the synchronous step in which
 $(x_1,x_1)$ was inserted for the first time in  $\quora_{p_1}$. 
Observe that $\C \subseteq M_1$. 
Furthermore, any process $p \in S(x_1)$ must also be in $M_1$ 
(i.e., $S(x_1) \subseteq M_1$). Also, $x_1=I(M_1)$, and,  
hence, $|x_1|=|M_1|$. Therefore, the only set $Q_1 \subseteq S(x_1)$ 
such that $I(Q_1)=x_1$ is $Q_1=M_1$.  We define $M_2$ similarly, 
and conclude that the only set $Q_2 \subseteq S(x_2)$ such that 
$I(Q_2)=x_2$ is $Q_2=M_2$. Since
$Q_1 \cap Q_2 \supseteq \C \not= \emptyset$, the safety property holds. 
%\end{itemize}
\end{proof}

%================================================================
\begin{figure}
\begin{lstlisting}
  $\labels_p \leftarrow \emptyset$;
  $\quora_p \leftarrow \emptyset$;
  for each synchronous step do
    broadcast $(\mathit{IDENT},id(p))$;
    wait for the messages sent in this synchronous step; 
    $mset_p \leftarrow$ multiset of identifiers received in $(\mathit{IDENT},-)$ messages;
    $\quora_p \leftarrow \quora_p \cup \{(mset_p,mset_p)\}$ ('\label{f2-7}')
    $\labels_p \leftarrow \labels_p \cup \{mset_p\}$; ('\label{f2-8}')
  end for;
\end{lstlisting}		
\caption{Algorithm to implement $\HS$ without knowledge of membership 
(code for process $p$)}
\label{Fig-HS}
\end{figure}

%========================================================================

\section{Solving  Consensus in Homonymous Systems}
\label{sec:FDs-based-consensus}

%We have developed two consensus algorithms for homonymous systems. 
%One algorithm implements Consensus in $\mathit{HAS}[t < n/2,\HO]$, that is,
%in an homonymous asynchronous system with reliable links, using the 
%failure detector $\HO$, 
%and when a majority of processes are correct. 
%The other algorithm implements Consensus in $\mathit{HAS}[\HO,\HS]$, 
%that is,
%in an homonymous asynchronous system with reliable links, using the 
%failure detector $\HO$ and $\HS$. 
%
%The algorithms are derived from the algorithm in Figure 4 of 
%\cite{DBLP:conf/aina/BonnetR10} and the algorithm of Figure 3 of 
%\cite{DBLP:conf/wdag/BonnetR10}, respectively, proposed for 
%anonymous systems. We adapt the algorithms for homonymous 
%systems and failure detectors $\HO$ and $\HS$.
%One key element of the adaptation is a Leaders' 
%Coordination Phase, in which the multiple leaders that a $\HO$ 
%failure detector allows try to agree
%in a common value to be proposed in a round. The properties of 
%$\HO$ guarantee that this eventually happen, which deals with 
%the issue of having multiple leaders.
%
%The algorithms and their proofs can be found in the appendix.

\begin{figure}
\begin{lstlisting}
operation propose($v_p$): 
  $est1_p$ := $v_p$; $r_p$:=$0$; ('\label{f3-2}')
  start Tasks T1 and T2;
  
Task T1
  repeat forever
    $r_p$:=$r_p +1$;  ('\label{f3-8}')
    // Leaders' Coordination Phase
    broadcast $(COORD, id(p),r_p, est1_p)$; ('\label{f3-10}')
    wait until ($D.\leader_p \neq id(p)) \lor$ ('\label{f3-11a}')
      ($D.\multiplicity_p$ messages $(COORD, id(p),r_p,-)$ received); ('\label{f3-11b}')
    if (some message $(COORD, id(p),r_p,-)$ received) then 
       $est1_p$:=$\min\{est_q :id(p)=id(q) \land$
                         $(COORD, id(q),r_p,est_q)$ received $\}$ end if;  ('\label{f3-13}')
    // Phase 0
    wait until $(D.\leader_p=id(p) \lor ((PH0, r_p,v)$ received); ('\label{f3-15}')
    if ($(PH0, r_p,v)$ received) then $est1_p$ :=$v$ end if; ('\label{f3-16}')
    broadcast$(PH0,r_p,est1_p)$; ('\label{f3-17}')
    // Phase 1
    broadcast$(PH1,r_p,est1_p)$;
    wait until $(PH1, r_p,-)$ received from $n-t$ processes;
    if (the same estimate $v$ received from $>n/2$ processes) then 
       $est2_p$:=$v$ ('\label{f3-22}')
    else
       $est2_p$:=$\bot$ ('\label{f3-24}')
    end if;
    // Phase 2
    broadcast$(PH2,r_p,est2_p)$;
    wait until $(PH2,r_p,-)$ received from $n-t$ processes;
    let $rec_p=\{est2:$ message $(PH2,r_p,est2)$ received $\}$;
    if $((rec_p=\{v\}) \wedge (v\not=\bot))$ then 
      broadcast $(DECIDE,v)$; return($v$) end if; ('\label{f3-30}')
    if $((rec_p=\{v,\bot\}) \wedge (v\not=\bot))$ then  $est1_p$:=$v$  end if; ('\label{f3-31}')
    if $(rec_p=\{\bot\})$ then skip  end if;
  end repeat;

Task T2
  upon reception of $(DECIDE,v)$ do
    broadcast $(DECIDE,v)$; return($v$)
\end{lstlisting}
\caption{Consensus algorithm in $\mathit{HAS}[t < n/2,\HO]$
(code for process $p$).
It uses detector $D \in \HO$.
}
\label{Fig-Cons-HOmega}
\end{figure}

%\subsection{Solving  Consensus in Homonymous Systems}
%\label{sec:FDs-based-consensus}

We present in this section two algorithms. One algorithm implements 
Consensus in $\mathit{HAS}[t < n/2,\HO]$, that is,
in an homonymous asynchronous system with reliable links, using the failure
detector $\HO$,  
and when a majority of processes are correct. 
The other algorithm implements Consensus in $\mathit{HAS}[\HO,\HS]$, that is,
in an homonymous asynchronous system with reliable links, using the failure
detector $\HO$ and $\HS$.  
%It is worth noting that in both cases we do not need to know the membership.

\subsection{Implementing Consensus in $\mathit{HAS}[t < n/2,\HO]$}
 Let us  consider $\mathit{HAS}[t < n/2,\HO]$ where  membership is unknown,
but the number of processes is known (that is, $n$). 
Let us assume a majority of correct processes (i.e., $t < n/2$).
We say  that a process $p$  is a leader, if  it is correct  and, after some
finite  time,  $D.\leader_q=id(p)$  permanently  for each  correct  process
$q$. By definition of $\HO$, there has to be at least one leader. 

%\paragraph{Brief explanation of the algorithm:}
The algorithm of Figure \ref{Fig-Cons-HOmega} is derived from the algorithm
in  Figure  4 of  \cite{DBLP:conf/aina/BonnetR10},  proposed for  anonymous
systems. This algorithm has been adapted for homonymous  
systems.  The  algorithm of  Figure  \ref{Fig-Cons-HOmega}  uses a  failure
detector  of class $\HO$  (instead of  $\AO$), and  a new  initial leaders'
coordination phase has been added. The  purpose of this initial phase is to
guarantee that, after a given round,  all leaders propose the same value in
each round. 

The  algorithm  works   in  rounds,  and  it  has   four  phases  (Leaders'
Coordination Phase, Phase 0, Phase 1 and Phase 2).  
Every  process $p$ begins  the Leaders'  Coordination phase  broadcasting a
$(COORD,id(p),r,est1_p)$ message. 
If process $p$ considers itself a leader (querying the failure detector $D$
of class $\HO$), it has to wait 
 until  to   receive  $(COORD,id(p),r,est1)$  messages  sent   by  all  its
homonymous  processes (also  querying  the failure  detector  $D$ of  class
$\HO$) (Lines~\ref{f3-11a}-\ref{f3-11b}).  After that, process  $p$ updates
its estimate $est1_p$ with the minimal value  
proposed among all its homonymous.  Note that eventually all its homonymous
will be leaders too. Hence, eventually  
all leaders will also choose the same minimal value in $est1$. 

In Phase 0, if process $p$  considers itself a leader (querying the failure
detector $D$ of class $\HO$) (Line~\ref{f3-15}), 
it   broadcast   a   $(PH0,r,est1_p)$   message  with   its   estimate   in
$est1_p$. Otherwise, process $p$ has to update its $est1_p$ waiting until a
$(PH0,r,est1_l)$ message is received from  one of the leaders processes $l$
(Lines~\ref{f3-15}-\ref{f3-16}).  
Note that after the Leaders' Coordination Phase, eventually each leader $l$
broadcast $(PH0,-,est1_l)$ messages with the same value in $est1_l$. 

%In Phase 1, every process $p$ broadcasts a $(PH1,r,est1_p)$ message with its estimate in $est1_p$. Then, it waits
% until the reception of $(PH1,r,est1)$ messages from a majority of process. If all $est1$ of the received messages are equal (for example $v$), process $p$ will take it updating its variable $est2_p$ to $v$ (Line 22). Otherwise, process $p$ will update its variable $est2_p$ to $\bot$ (Line 24).  Note that when this Phase 1 finishes, the number of possible different values in $est2_p$ are only two: $v$ or $\bot$.

%In Phase 2,  every process $p$ broadcasts a $(PH2,r,est2_p)$ message with its estimate in $est2_p$. Then, it waits
% until the reception of $(PH2,r,est2)$ messages from a majority of processes. If all $est2$ of the received messages 
% are the same value $v$ different from $\bot$, process $p$ will decide $v$ (Line 30). Otherwise, if some of the estimate 
% $est2$ received from messages are the value $v$ different from $\bot$, process $p$ will take it updating its variable $est1_p$ to $v$ (Line 31) in order to propose it in the next round $r+1$. With this, if a majority of processes
% ($p$ not included) decides in this round $r$ a value $v$, it will ensure that the process $p$ will propose this same value $v$ in the next round $r+1$. Finally, if all estimate $est2$ received from messages are the value $\bot$, nothing is performed.  
% 
%Finally, Task 2 implements a reliable broadcast needed to propagate a decided value $v$ from one process to the rest of processes of the system. 

The  rest of  the algorithm  is similar  to the  algorithm in  Figure  4 of
\cite{DBLP:conf/aina/BonnetR10}.  We  omit  further  details due  to  space
restrictions. 
The following lemmas are the key  of the correctness of the algorithm. They
show that, even having multiple  leaders, these will eventually converge to
propose the same value at each round. 
 
\begin{lemma} 
\label{no-block}
No correct process blocks forever in the Leaders' Coordination Phase.
\end{lemma} 
\begin{proof}
The only line in  which processes can block in Lines~\ref{f3-8}-\ref{f3-13}
is in Lines~\ref{f3-11a}-\ref{f3-11b}. A correct process that is not leader
does not  block permanently  in these lines,  because eventually  the first
part of  
the wait  condition is  satisfied. Let us  assume, for  contradiction, that
some leader  blocks permanently in  Line~\ref{f3-11b}. Let us  consider the
smallest round $r$  in which some leader $p$ blocks.  By definition of $r$,
each leader  $q$ eventually reaches  round $r$, and  (even if it  blocks in
$r$)   broadcasts  $(COORD,   id(q),  r,   -)$,  where   $id(q)=id(p)$,  in
Line~\ref{f3-10}. (Observe that all processes send $(COORD,-,-,-)$ messages  
in Line~\ref{f3-10}, even  if they do not consider  themselves as leaders.)
Eventually, all these messages are delivered to $p$  
and $D.\multiplicity_p$  is permanently the  number of leaders.  Hence, the
second part  of the wait condition (Line~\ref{f3-11b})  is satisfied. Thus,
$p$ is not blocked anymore, and, therefore, we reach a contradiction. 
\end{proof}

\begin{lemma} 
\label{same-value}
There  is a  round  $r$ such  that  at every  round $r'  >  r$ all  leaders
broadcast the same value in Phase 0 of round $r'$. 
\end{lemma} 
\begin{proof}
Eventually all leaders  broadcast the same value because  after some round,
all leaders start Phase 0 with the same value in $est1$. 
Consider  a time  $\tau$ when  all faulty  processes have  crashed  and the
failure detector $D$ is stable (i.e.,  
$\forall \tau' \ge \tau, \forall p \in \C$, $D.\leader_p^{\tau'}=\ell$,
 being $\ell \in I(\C)$, and $D.\multiplicity_p^{\tau'}=mult_{I(C)}(\ell)$). 
Let $r$ be  the largest round reached by any process  at time $\tau$. Then,
for any round $r'>r$, all 
leaders $p$ have the same estimate $est1_p$ at the beginning of the Phase 0
of round $r'$ (Line ~\ref{f3-15}), or  there has been a decision in a round
smaller than $r'$. To prove this, let us assume that no decision is reached
in a round smaller than $r'$. 
Then, since  the leaders do  not block forever  in any round  (see previous
paragraph  1),  they execute  Line~\ref{f3-10}  in  round  $r'$. Since  the
failure detector is stable, they also  wait for the second part of the wait
condition of  Lines~\ref{f3-11a}-\ref{f3-11b} (since the first  part is not
satisfied). 
When any  leader $p$ executes the  Leaders' Coordination Phase  of $r'$, it
blocks    in     Lines~\ref{f3-11a}-\ref{f3-11b}    until    it    receives
$D.\multiplicity_p$ messages  from the other  leaders. By the  stability of
the  $\HO$ failure  detector, $D.\multiplicity_p$  is the  exact  number of
leaders. Also,  from the  definition of $\tau$  and $r$, no  faulty process
with  identifier $D.\leader_p$  is alive  and  all the  messages they  sent
correspond to rounds smaller than $r'$. 
Hence, each  leader $p$ will  wait to receive  messages from all  the other
leaders and  will set $est1_p$ to the  minimum from the same  set of values
(Line~\ref{f3-13}). 
\end{proof}

%The proof of the following theorem can be found in Appendix~\ref{proof:thm1}.

\begin{theorem}
\label{thm:consensus1}
The algorithm of Figure~\ref{Fig-Cons-HOmega} 
solves consensus in $\mathit{HAS}[t<n/2,\HO]$.
% (homonymous asynchronous system 
%with a majority of correct processes using a failure detector of class $\HO$). 
\end{theorem}
\begin{proof}
From the definition of Consensus, 
it is enough to prove the following properties.

%\begin{itemize}
%\item 
\emph{Validity.} 
The variable  $est1$ is  initialized with a  value proposed by  its process
(Line~\ref{f3-2}). The value of  $est1$ may be updated in Lines~\ref{f3-13}
or ~\ref{f3-16} with  values of $est1$ broadcasted by  other processes. The
variable $est2$ is initialized  and updated with  $est1$ (Line~\ref{f3-22})
or  $\bot$  (Line~\ref{f3-24}). The  value  of  $est1$  may be  updated  in
Line~\ref{f3-31} with values of  $est2$ (different from $\bot$) broadcasted
by other processes.  The value decided in Line~\ref{f3-30}  is the value of
$est2$  that was  broadcasted by  some process.  As it  is not  possible to
decide the value  $\bot$ (Line~\ref{f3-30}), then the value  decided has to
be  one  of the  values  proposed by  the  processes.  Then,  the  validity
property holds.   

%\item 
\emph{Agreement.} 
Identical to the agreement property of 
Figure 4 of \cite{DBLP:conf/aina/BonnetR10},

%\item 
\emph{Termination.}
From Lemmas \ref{no-block} and  \ref{same-value}, after some round $r$, all
leaders  
hold the  same value  $v$ in $est1$  when they  start executing Phase  0 of
round  $r'$ (Line~\ref{f3-15}),  and  they broadcast  this  same value  $v$
(Line~\ref{f3-17}). Note that it is the same situation as having only  
one leader with value $v$ stored  in $est1$ when Phase 0 is reached. Hence,
as Phase 0 starts in the same conditions as in the algorithm of Figure 4 of
\cite{DBLP:conf/aina/BonnetR10}, the  same proof can  be used to  prove the
termination property. 
%\end{itemize}
\end{proof}

\subsection{Implementing Consensus in $\mathit{HAS}[\HO,\HS]$}
Figure      \ref{Fig-Cons-HO-Sigma}       implements      Consensus      in
$\mathit{HAS}[\HO,\HS]$. Note  that it is  a variation of the  algorithm of
Figure  3 of  \cite{DBLP:conf/wdag/BonnetR10} where,  like in  the previous
case, we have  added a preliminary phase as a  barrier such that homonymous
leaders  eventually  ``agree"  in  the  same  estimation  value  $est1$  to
propose. Once  this issue has been  solved (as was proven  for the previous
algorithm), the use that this algorithm makes of the 
failure detector $\HS$ is very similar to the use the algorithm of Figure 3
of \cite{DBLP:conf/wdag/BonnetR10} makes of the $\AS$ failure detector.

\begin{figure}
\begin{lstlisting}
operation propose($v_p$):
  $est1_p$ := $v_p$; $r_p$ := $0$;
  start Tasks T1 and T2;
  
  Task T1
   repeat forever
    $r_p$:=$r_p +1$; 
    // Leaders' Coordination Phase
    broadcast $(COORD, id(p),r_p, est1_p)$;
    wait until ($D1.h\_leader_p \neq id(p)) \lor$ 
      ($D1.h\_multiplicity_p$ messages $(COORD, id(p),r_p,-)$ received);
    if (some message $(COORD, id(p),r_p,-)$ received) then 
      $est1_p$:=$\min\{est_q :id(p)=id(q) \land$ 
                      $(COORD, id(q),r_p,est_q)$ received $\}$ end if;
    // Phase 0
    wait until $(D1.h\_leader_p=id(p) \lor ((PH0, r_p,v)$ received);
    if ($(PH0, r_p,v)$ received) then $est1_p$ :=$v$ end if; 
    broadcast$(PH0,r_p,est1_p)$;
    // Phase 1
    $sr_p$:= $1$; $current\_labels_p$ := $D2.h\_labels_p$; 
    broadcast $(PH1,id(p),r_p, sr_p, current\_labels_p, est1_p)$;
    repeat             ('\label{f4-20}')
      if ($(PH2,-,r_p ,-,-, est2)$ received) then 
        $est2_p$ := $est2$; exit inner repeat loop end if;              ('\label{f4-21}')
      if $((\exists (x, mset) \in D2.h\_quora_p) \land (\exists sr \in N) \land$              ('\label{f4-22}')
         $(\exists$ set $M$ of messages $(PH1,-,r_p , sr, -,-))$, such that, 
         $(\forall (PH1,-,-,-, cl ,-) \in M, x \in cl) \land$
         $(mset=\{i:(PH1,i,-,-,-,-) \in M\}))$ then             ('\label{f4-23}')
         if (all msgs in $M$ contain the same estimate $v$) then 
             $est2_p$:=$v$ else $est2_p$:=$\bot$ end if;
         exit inner repeat loop;
      else if ($current\_labels_p \not= D2.h\_labels_p) \lor$             ('\label{f4-26}')
            $((PH1,-,r_p, sr,-,-)$ received with $sr > sr_p$) then
            $sr_p$:=$sr_p + 1$; $current\_labels_p$:=$D2.h\_labels_p$; 
            broadcast $(PH1,id(p),r_p, sr_p, current\_labels_p, est1_p)$
          end if             ('\label{f4-29}')
      end if
    end repeat;             ('\label{f4-31}')
    // Phase 2 
    $sr_p$:=$1$; $current\_labels_p$:=$D2.h\_labels_p$;  
    broadcast $(PH2,id(p),r_p, sr_p, current\_labels_p, est2_p)$;                ('\label{f4-34}')
    repeat              ('\label{f4-35}')
      if ($(COORD,-,r_p + 1,-)$ received) then 
         exit inner repeat loop end if;
      if $((\exists (x, mset) \in D2.h\_quora_p) \land (\exists sr \in N) \land$              ('\label{f4-37}')
         $(\exists$ set $M$ of messages $(PH2,-,r_p , sr, -,-))$, such that, 
         $(\forall (PH2,-,-,-, cl ,-) \in M, x \in cl) \land$
         $(mset=\{i:(PH2,i,-,-,-,-) \in M\}))$ then             ('\label{f4-38}')
         let $rec_p$ = the set of estimates contained in $M$;
         if ($(rec_p = \{v\}$) $\wedge$ ($v\not=\bot)$) then 
           broadcast $(DECIDE,v)$; return(v) end if;             ('\label{f4-40}')
         if ($(rec_p = \{v, \bot\})$ $\wedge$ ($v\not=\bot)$) then $est1_p$:=$v$ end if;
         if $(rec_p = \{\bot\}$) then skip end if;
         exit inner repeat loop
      else if $((current\_labels_p \not= D2.h\_labels_p) \lor$
             $((PH2,-,r_p, sr,-,-)$ received with $sr > sr_p))$ then
             $sr_p$:=$sr_p + 1$; $current\_labels_p$:=$D2.h\_labels_p$; 
             broadcast $(PH2,id(p),r_p, sr_p, current\_labels_p, est2_p)$              ('\label{f4-46}')
          end if
      end if
    end repeat              ('\label{f4-49}')
  end repeat;
     
  Task T2
   upon reception of $(DECIDE,v)$ do 
     broadcast $(DECIDE,v)$; return(v)
\end{lstlisting}
\caption{Consensus algorithm in $\mathit{HAS}[\HO,\HS]$ (code for process $p$). It uses detectors $D1 \in \HO$ and $D2 \in \HS$.}
\label{Fig-Cons-HO-Sigma}
\end{figure}    

\begin{lemma}
\label{no-block-repeats}
No correct process blocks forever in the repeat loops of Phases 1 and 2.
\end{lemma}
\begin{proof}
Note that if a correct  process decides (Line~\ref{f4-40}), then the claims
follows.  
Consider the repeat loop of Phase 1 (Lines~\ref{f4-20}-\ref{f4-31}). Let us
assume that some correct process is blocked forever in  
this loop.  Then, let us  consider the first  round $r$ in which  a correct
process  blocks forever  in $r$.  Hence, all  correct processes  must block
forever in the same loop in  round $r$. Otherwise some process broadcasts a
message  
$(PH2,-,r,-,-,-)$, and from Line~\ref{f4-21} no correct process would block
forever in this loop of round $r$. 
Let us consider a correct process $p$, and the pair $(x,m)$ that guarantees
the liveness property for $p$. Then, there 
is a time in which $(x,m) \in D2.\quora_p$ and every correct process $q$ in
$S(x) \cap \C$ has $x \in D2.\labels_q$. 
Note that, from Lines~\ref{f4-26}-\ref{f4-29}, every change in the variable
$D2.\labels$ of  a process creates a  new sub-round, and  that all processes
broadcast   their    current   value   of   $D2.\labels$    in   each   new
sub-round. Therefore,  
eventually, $p$  will receive  messages $(PH1,-,r,sr,cl,-)$ from  all these
processes   such    that   $x   \in   cl$.   Hence,    the   condition   of
Lines~\ref{f4-22}-\ref{f4-23} is  satisfied, and $p$ will exit  the loop of
Phase 1. The argument for the repeat loop of  
Phase 2 is verbatim.
\end{proof}

\begin{lemma}
\label{no-different-dec}
No two processes decide different values in the same round.
\end{lemma}
\begin{proof}
Let us assume that processes $p_1$  and $p_2$ decide values $v_1$ and $v_2$
in sub-rounds  $sr_1$ and  $sr_2$, respectively, of  the same round  $r$ (in
Line~\ref{f4-40}).   Let  $(x_1,   m_1)$   and  $M_1$   be   the  pair   in
$D2.\quora_{p_1}$ and  the set  of messages that  satisfy the  condition of
Lines~\ref{f4-37}-\ref{f4-38} for $p_1$. Since for each message 
$(PH2,-, r, sr_1, cl ,-) \in M_1$,  it holds that $x_1 \in cl$, if $Q_1$ is
the set  of senders of the messages  in $M_1$, we have  that $Q_1 \subseteq
S(x_1)$. Additionally, $m_1=\{i:(PH2,i,-,-,-,-) \in M_1\}=I(Q_1)$. 
We can define $(x_2, m_2)$ and  $M_2$ analogously for $p_2$. Then, from the
Safety Property of $\HS$, $Q_1 \cap  Q_2 \not= \emptyset$. Let $p_l \in Q_1
\cap   Q_2$.   Then,   process   $p_l$   must   have   broadcast   messages
$(PH2,id(p_l),r,sr_1,-,v_1)$        and        $(PH2,id(p_l),r,sr_2,-,v_2)$
(Lines~\ref{f4-34}  and \ref{f4-46}).  Since the  estimate  $est2_{p_l}$ of
$p_l$   does   not   change   between   sub-rounds   (inner   repeat   loop,
Lines~\ref{f4-35}-\ref{f4-49}),  it  must  hold  that $v_1=v_2$.  From  the
condition of  Line~\ref{f4-40}, $rec_{p_1}=\{v_1\}$ in  sub-round $sr_1$ and
$rec_{p_2}=\{v_2\}$ in sub-round $sr_2$,  and both processes decide the same
value. Hence, no two processes decide different values in the same round. 
\end{proof}

 \begin{theorem}
 \label{thm:consensus2}
The  algorithm   of  Figure~\ref{Fig-Cons-HO-Sigma}  solves   consensus  in
 $\mathit{HAS}[\HO,\HS]$. 
% (homonymous asynchronous system 
%using failure detectors of classes $\HO$ and $\HS$). 
\end{theorem}
\begin{proof}
The  proof  of this  theorem  is  similar to  the  proof  of  Theorem 5  of
 \cite{TR-DBLP:conf/wdag/BonnetR10}         (full         version        of
 \cite{DBLP:conf/wdag/BonnetR10}), with the following changes. Observe that
 the Leaders' Coordination Phase and  
 Phase   0  of   the  algorithms   in  Figures   \ref{Fig-Cons-HOmega}  and
 \ref{Fig-Cons-HO-Sigma}  are the  same. Hence,  Lemmas  \ref{no-block} and
 \ref{same-value}    also    apply    to    the   algorithm    of    Figure
 \ref{Fig-Cons-HO-Sigma}.  Then,  
 the  termination  property   can  be  proven  in  a   similar  way  as  in
 \cite{TR-DBLP:conf/wdag/BonnetR10} (Lemmas  1 and 2), but  using those two
 Lemmas   \ref{no-block}   and   \ref{same-value}   together   with   Lemma
 \ref{no-block-repeats}.  The  proof  of  the agreement  property  is  also
 similar to  Lemma 3 of \cite{TR-DBLP:conf/wdag/BonnetR10}  but using Lemma
 \ref{no-different-dec}. 
\end{proof}

%Since   classes  $\MO$  and   $\AO$  are   equivalent,  the   algorithm  of
%Figure~\ref{Fig-Cons-HO-Sigma}  can be transformed  into an  algorithm that
%solves consensus in $AAS[\MO,\HS]$. 

The algorithm of Figure~\ref{Fig-Cons-HO-Sigma} can be easily transformed into an algorithm that solves consensus in 
$\mathit{AAS}[\AO,\HS]$ (an anonymous system with detectors $\AO$ and $\HS$).
For that, given a failure detector $D3 \in \AO$, it is enough to remove the Leaders' Coordination Phase, 
and in Phase~0 to replace $(D1.h\_leader_p=id(p))$ by $(D3.a\_leader_p)$. The resulting Phase~0 is the same as Phase~1 in the algorithm of Figure 3 of \cite{DBLP:conf/wdag/BonnetR10}, and has the same properties.

%========================================================================
%\section{Conclusion}
%\label{sec:conclusion}

%In this paper we have studied the Consensus problem in a new distributed environment with homonymous processes. We have defined new failure detectors for this homonymous model (called $\HO$, $\HP$ and $\HS$).  
%We have also studied the relations among the failure detector classes $\Sigma$ and its versions for homonymous systems (denoted by $\HS$), and 
%for anonymous systems (denoted by $\AS$).
%We have also shown that $\HO$, $\HP$ and $\HS$ can be implemented in homonymous systems with different synchrony requirements (and in all cases without initial knowledge of the membership). 
%Interestingly, our class $\HO$ can be implemented with partial synchrony, while the analogous class $\AO$ defined for anonymous systems can not be implemented (even in synchronous systems).
%Finally, we have shown that $\HO$ (and a majority of correct processes), and $\langle \HO,\HS\rangle$ failure detectors can be used to implement Consensus in homonymous asynchronous systems (even without initial knowledge of the membership).

%\bibliography{refs}
%\bibliographystyle{plain}

%========================================================================
\newpage
\appendix

\input{appendix-reductions}

% - La transformacion de nuestro $h\_sigma$ en $a\_sigma$ de Michel 
%hay que estudiarla porque no es evidente.
%
%- Se puede relajar la propiedad Liveness por otra en la que se 
%cambie $\forall p \in Correct$ por $\exists p \in Correct$

\end{document}

%% file: FD-relations-latex.pstex_t
\begin{picture}(0,0)%
\includegraphics{FD-relations-latex.pstex}%
\end{picture}%
\setlength{\unitlength}{3158sp}%
\begingroup\makeatletter\ifx\SetFigFont\undefined%
\gdef\SetFigFont#1#2#3#4#5{%
  \reset@font\fontsize{#1}{#2pt}%
  \fontfamily{#3}\fontseries{#4}\fontshape{#5}%
  \selectfont}%
\fi\endgroup%
\begin{picture}(4658,2269)(769,-2409)
\put(4055,-1169){\makebox(0,0)[b]{\smash{{\SetFigFont{10}{12.0}{\rmdefault}{\mddefault}{\updefault}{\color[rgb]{0,0,0}$\NAP$}%
}}}}
\put(2481,-1147){\makebox(0,0)[b]{\smash{{\SetFigFont{10}{12.0}{\rmdefault}{\mddefault}{\updefault}{\color[rgb]{0,0,0}$\HP$}%
}}}}
\put(5149,-1165){\makebox(0,0)[b]{\smash{{\SetFigFont{10}{12.0}{\rmdefault}{\mddefault}{\updefault}{\color[rgb]{0,0,0}$\HP$}%
}}}}
\put(1838,-422){\makebox(0,0)[b]{\smash{{\SetFigFont{10}{12.0}{\rmdefault}{\mddefault}{\updefault}{\color[rgb]{0,0,0}$\Sigma$}%
}}}}
\put(2659,-428){\makebox(0,0)[b]{\smash{{\SetFigFont{10}{12.0}{\rmdefault}{\mddefault}{\updefault}{\color[rgb]{0,0,0}$\HS$}%
}}}}
\put(1009,-428){\makebox(0,0)[b]{\smash{{\SetFigFont{10}{12.0}{\rmdefault}{\mddefault}{\updefault}{\color[rgb]{0,0,0}$\AS$}%
}}}}
\put(2714,-1834){\makebox(0,0)[b]{\smash{{\SetFigFont{10}{12.0}{\rmdefault}{\mddefault}{\updefault}{\color[rgb]{0,0,0}$\HO$}%
}}}}
\put(1894,-1834){\makebox(0,0)[b]{\smash{{\SetFigFont{10}{12.0}{\rmdefault}{\mddefault}{\updefault}{\color[rgb]{0,0,0}$\OO$}%
}}}}
\put(1074,-1842){\makebox(0,0)[b]{\smash{{\SetFigFont{10}{12.0}{\rmdefault}{\mddefault}{\updefault}{\color[rgb]{0,0,0}$\AO$}%
}}}}
\put(1544,-1159){\makebox(0,0)[b]{\smash{{\SetFigFont{10}{12.0}{\rmdefault}{\mddefault}{\updefault}{\color[rgb]{0,0,0}$\NAP$}%
}}}}
\put(1931,-2331){\makebox(0,0)[b]{\smash{{\SetFigFont{10}{12.0}{\rmdefault}{\mddefault}{\updefault}{\color[rgb]{0,0,0}$\mathit{AS}[\emptyset]$ system model}%
}}}}
\put(3533,-360){\makebox(0,0)[b]{\smash{{\SetFigFont{10}{12.0}{\rmdefault}{\mddefault}{\updefault}{\color[rgb]{0,0,0}$\AS$}%
}}}}
\put(4469,-362){\makebox(0,0)[b]{\smash{{\SetFigFont{10}{12.0}{\rmdefault}{\mddefault}{\updefault}{\color[rgb]{0,0,0}$\HS$}%
}}}}
\put(4689,-1905){\makebox(0,0)[b]{\smash{{\SetFigFont{10}{12.0}{\rmdefault}{\mddefault}{\updefault}{\color[rgb]{0,0,0}$\HO$}%
}}}}
\put(4572,-2317){\makebox(0,0)[b]{\smash{{\SetFigFont{10}{12.0}{\rmdefault}{\mddefault}{\updefault}{\color[rgb]{0,0,0}$\mathit{AAS}[\emptyset]$ system model}%
}}}}
\end{picture}%

%% file: appendix-reductions.tex
\subsection{Reductions between Failure Detectors}
\label{appendixred-FD}

\subsubsection{From $\Sigma$ to $\HS$}

\begin{figure}
\begin{lstlisting}
Init
  $\labels_p \leftarrow \{s: (s \subseteq I(\Pi)) \wedge (id(p) \in s)\}$;			('\label{a1-l2}')
  $\quora_p \leftarrow \emptyset$;											('\label{a1-l3}')
  repeat forever
    $q \leftarrow D.trusted_p$; 												('\label{a1-l5}')
    $\quora_p  \leftarrow \quora_p \cup \{(q,q)\}$;								('\label{a1-l6}')
  end repeat; 
\end{lstlisting}		
\caption{Algorithm to transform $D \in \Sigma$ to $\HS$ with initial knowledge of membership (code for process $p$).}
\label{Fig-S-to-HS-si-mship}
\end{figure}

\begin{figure}
\begin{lstlisting}
Init
  $\labels_p \leftarrow \emptyset$;					('\label{a2-l2}')
  $\quora_p \leftarrow \emptyset$;					('\label{a2-l3}')
  $mship_p \leftarrow \emptyset$;					('\label{a2-l4}')
  start tasks T1 and T2;
Task T1
  repeat forever
    broadcast $(\mathit{IDENT}, id(p))$;				('\label{a2-l8}')
    $q \leftarrow D.trusted_p$; 						('\label{a2-l9}')
    $\quora_p  \leftarrow \quora_p \cup \{(q,q)\}$;		('\label{a2-l10}')
  end repeat;

Task T2	
  upon reception of $(\mathit{IDENT}, i)$ do
    $mship_p \leftarrow mship_p \cup \{i\}$
    $\labels_p \leftarrow \{s: (s \subseteq mship_p) \wedge (id(p) \in s)\}$;  
\end{lstlisting}		
\caption{Algorithm to transform $D \in \Sigma$ to $\HS$ without initial knowledge of membership (code for process $p$).}
\label{Fig-S-to-HS-no-mship}
\end{figure}

We prove that, if identifiers are unique, a detector of class $\HS$ can be obtained from any detector $D$ of class $\Sigma$.

\begin{theorem}
\label{HS-Sigma}
A failure detector of class $\HS$ can be obtained from any detector $D$ of class $\Sigma$ in a system with unique identifiers, under either one of the following conditions:
\begin{enumerate} 
\item without any communication if every process initially knows the membership $I(\Pi)$, or 
\item in system $AS[\Sigma]$ (the membership does not need to be known initially). 
\end{enumerate} 
\end{theorem}
\begin{proof}
Let $D.trusted_p$ be the variable of $\Sigma$ failure detector $D$ at process $p$. 
Figures \ref{Fig-S-to-HS-si-mship} and \ref{Fig-S-to-HS-no-mship} present the algorithms to transform $D$ into a 
failure detector of class $\HS$ in Cases 1 and 2, respectively.
In both cases, at each process $p$ initially
$\quora_p \leftarrow \emptyset$, and infinitely often this variable is updated with the following sentences: 
$q \leftarrow D.trusted_p$, and $\quora_p  \leftarrow \quora_p \cup \{(q,q)\}$. In Case 1, initially every process $p$ sets 
$\labels_p \leftarrow \{s: (s \subseteq I(\Pi)) \wedge (id(p) \in s)\}$ and it never changes it in the run. In Case 2, every process $p$
initially sets $\labels_p \leftarrow \emptyset$, and repeatedly broadcasts a message $\mathit{IDENT}(id(p))$. Process $p$ also has a variable $mship_p$ initially set
to $mship_p \leftarrow \emptyset$. After receiving a message $\mathit{IDENT}(i)$, process $p$ updates 
$mship_p \leftarrow mship_p \cup \{i\}$, and $\labels_p \leftarrow \{s: (s \subseteq mship_p) \wedge (id(p) \in s)\}$. 

We prove now the properties of $\HS$:

\begin{itemize}
\item Validity.
Since $\quora_p$ is a set, and the elements included in it are of the form $(q,q)$ (see Line 5 in Figure \ref{Fig-S-to-HS-si-mship}, and Line 10 in Figure \ref{Fig-S-to-HS-no-mship}) there can not be two pairs with the same label.

\item Monotonicity.
The monotonicity of $\labels_p$ in Figure \ref{Fig-S-to-HS-si-mship} is obvious because it is initialized in Line 2 and never changes. With respect to Figure \ref{Fig-S-to-HS-no-mship}, $\labels_p$ is initially empty, and it is related with the set $mship_p$, such that if $mship_p$ grows then $\labels_p$ either grows or remains the same. Hence $\labels_p$ never decreases because $mship_p$ never decreases (see Line 15 in Figure \ref{Fig-S-to-HS-no-mship}). 
The monotonicity of $\quora_p$ in Figures \ref{Fig-S-to-HS-si-mship} and \ref{Fig-S-to-HS-no-mship} follows from the fact that $\quora_p$ is initially empty, and any element $(q,q)$ included in it is never removed.

\item Liveness.
Consider any correct process $p$. In Figure \ref{Fig-S-to-HS-no-mship},  eventually, $\C \subseteq mship_p$ permanently (from the exchange of $\mathit{IDENT}$ messages and Line 15 of Figure \ref{Fig-S-to-HS-no-mship}). Then, in both algorithms eventually 
$\{s: (s \subseteq I(\C)) \wedge (id(p) \in s)\}  \subseteq \labels_p $ permanently (from Line 2 in Figure \ref{Fig-S-to-HS-si-mship}, and Line 16 in Figure \ref{Fig-S-to-HS-no-mship}). Hence, there is a time $\tau$ after which, for every set $s \subseteq I(\C)$, $I(S(s)) = s$ and $S(s) \subseteq \C$.

The Liveness property of $\Sigma$ guarantees that, at some time $\tau' \ge \tau$, the variable $q$ is assigned  a set  $s$ that contains only correct processes and $(s,s)$ will be included in $\quora_p$ after that. Therefore, there is a time after which $\quora_p$ contains $(s,s)$ permanently (from monotonicity).  Since $s \subseteq I(S(s) \cap \C)=I(S(s))=s$, the property follows.
 
\item Safety.
Consider two pairs $(x_1,m_1) \in \quora_{p_1}^{\tau_1}$ and $(x_2,m_2) \in \quora_{p_2}^{\tau_2}$, for any $p_1, p_2 \in \Pi$ and any $\tau_1, \tau_2 \in N$. From the management of the $\quora$ variables (Lines \ref{a1-l3}, \ref{a1-l5}, and \ref{a1-l6} in Figure \ref{Fig-S-to-HS-si-mship}, and Lines \ref{a2-l3}, \ref{a2-l9}, and \ref{a2-l10} in Figure \ref{Fig-S-to-HS-no-mship}), we have that $m_1$ and $m_2$ are values taken from $D.trusted_{p_1}$ and $D.trusted_{p_2}$, respectively. Hence, the sets $m_1$ and $m_2$ must intersect from the Safety property of the $\Sigma$ failure detector $D$. Then, if $I(Q_1) = m_1$ and  $I(Q_2) = m_2$, given that we are in a system with unique identifiers, $Q_1$ and $Q_2$ must intersect.
\end{itemize}
\end{proof}

\subsubsection{From $\HS$ to $\Sigma$}

\begin{figure}
\begin{lstlisting}
Init
  $alive_p \leftarrow$ empty list;
  start Tasks T1 and T2;
Task T1
  repeat forever
    broadcast $(\mathit{ALIVE}, id(p))$;
  end repeat;

Task T2	
  upon reception of $(\mathit{ALIVE}, i)$ do
    if $i \in alive_p$ then move $i$ to the first position of $alive_p$
    else insert $i$ in the first position of $alive_p$
    end if;
\end{lstlisting}		
\caption{Algorithm to implement a failure detector of class $\Xi$ without initial knowledge of membership in AS[$\emptyset$] (code for process $p$).}
\label{Fig-Xi}
\end{figure}

%============================================================================
\begin{figure*}
\begin{lstlisting}
Init
  start Tasks T1 and T2;
Task T1
  repeat forever
    broadcast $(\mathit{LABELS}, id(p), D.\labels_p)$;
    if $\exists (x,m) \in D.\quora_p: (idents_p[x]$ has been created$) \wedge (m \subseteq idents_p[x])$ then
      let $candidates_p = \{ m : ((x,m) \in D.\quora_p) \wedge (idents_p[x]$ has been created$) \wedge (m \subseteq idents_p[x])\}$;
      $trusted_p \leftarrow$ any $m \in candidates_p$ with smallest $\max_{i \in m} rank(i, X.alive_p)$;
    end if;
  end repeat;

Task T2	
  upon reception of $(\mathit{LABELS}, i,\ell)$ do
    foreach $x \in \ell$ do 
      if $idents_p[x]$ has not been created then create $idents_p[x] \leftarrow \emptyset$ end if;
      $idents_p[x] \leftarrow idents_p[x] \cup \{i\}$;
    end foreach;
\end{lstlisting}		
\caption{Algorithm to transform $D \in \HS$ to $\Sigma$ in a system with unique identifiers, but without initial knowledge of membership (code for process $p$). The algorithm uses a failure detector $X$ of class $\Xi$.}
\label{Fig-HS-to-S-no-mship}
\end{figure*}

We define now a new class of failure detector that will be used for reductions between the above failure detector classes.
While the service provided by this detector has been already used \cite{DBLP:conf/podc/Zielinski08,DBLP:conf/wdag/BonnetR10}, it was never formally
defined. The new failure detector class, denoted $\Xi$, will only be defined for systems with unique identifiers, i.e., non homonymous.

\begin{definition}
A failure detector of class $\Xi$ provides each process  $p \in \Pi$, in a system with unique process identifiers, with a variable $alive_p$ which contains a (sorted) list of process identifiers. 
Any failure detector of class $\Xi$ must satisfy the following property:
\begin{itemize}
% \item[Uniqueness] 
% The same identifier cannot appear more than once in the list $alive_p$.
 \item Liveness.
 Eventually, the identifiers of the correct processes 
 %$I(\C)$ (observe that $I(\C)$ is a set) 
 are permanently in the first positions of $alive_p$. More formally,
 let $rank(i, alive_p^{\tau})$ denote the position (starting from $1$) of process identifier $i$ in $alive_p^{\tau}$ 
 (with $rank(i, alive_p^{\tau})=\infty$ if $i \notin alive_p^{\tau}$). Then,
 $\forall p \in \C , \exists \tau \in N: \forall \tau' \ge \tau, \forall q \in \C, rank(id(q), alive)_p^{\tau'} \leq |\C|.$
 \end{itemize}
\end{definition}

Observe that 
%$alive_p$ may never contain the identifier of all the processes in the system, that 
the position of the same identifier can be different at different processes, and can vary over time in the same process.
From the algorithm of Figure~\ref{Fig-Xi}, we obtain the following lemma.

\begin{lemma}
\label{l-Xi}
A failure detector of class $\Xi$ can be implemented in AS[$\emptyset$] (an asynchronous system with unique identifiers), even when the membership is not known initially.
\end{lemma}
\begin{proof}
For each process $q \in \C$, eventually some message $\mathit{ALIVE}(id(q))$ will be received at each process $p \in \C$. Then $id(q)$ will be included in $alive_p$ and never removed after that. Given any faulty process $r$, $p$ will stop receiving messages from $r$ by some time $\tau$. Then, after $\tau$ process $p$ will never receive a message $\mathit{ALIVE}(id(r))$ and $id(r)$ will never be moved to (inserted in) the first position of $alive_p$. However, after $\tau$, eventually $p$ will receive messages $\mathit{ALIVE}(id(q))$ from each process $q \in \C$, and each identifier $id(q)$ will be moved to (or inserted in) the first position of $alive_p$. Then, there is some time $\tau' > \tau$ such that, at all times $\tau'' > \tau'$, $rank(id(q), alive_p^{\tau''}) < rank(id(r), alive_p^{\tau''})$. Since this holds for all $p,q \in \C$ and all $r \notin \C$, the claim follows.
\end{proof}

We now show, using the algorithm of  Figure~\ref{Fig-HS-to-S-no-mship}, that $\Sigma$ can be obtained from $\HS$ without initial knowledge of the membership. 

\begin{theorem}
\label{Sigma-HS}
A failure detector of class $\Sigma$ can be obtained from any detector $D$ of class $\HS$ in AS[$\HS$] (an asynchronous system with unique identifiers), even when the membership is not known initially.
\end{theorem}
\begin{proof}
From Lemma~\ref{l-Xi}, we can have a failure detector of class $\Xi$ in an asynchronous system. The logic of the algorithm of Figure~\ref{Fig-HS-to-S-no-mship} is somewhat similar to that of the algorithm in Figure 2 in \cite{DBLP:conf/wdag/BonnetR10}. The condition in Line~6 guarantees that the variable $trusted _p$ is assigned a set of identifiers $m$ only if $(x,m)$ is in $\quora_p$, and every process $q$ whose identifier is in $m$ has $x$ in its set $\labels_q$ (from the management of the sets $idents_p$). Combining this condition with the safety property of $\HS$ we guarantee the safety property of $\Sigma$. The liveness property of $\Sigma$ holds from the liveness property of $\HS$, the choice of $m$ done in Line~8, and the properties of the failure detector class $\Xi$ as follows. If $p \in \C$, from the liveness of $\HS$, eventually every time Line~8 is executed, there is some $m \in candidates_p$ with only correct processes. If the failure detector $X$ of class $\Xi$ has already
all the correct processes in the lowest ranks of $X.alive_p$ (which eventually happens from its liveness property), then any set $m$ in $candidates_p$, whose largest rank in $X.alive_p$ is minimal, contains only correct processes (which yields the liveness of $\Sigma$).
\end{proof}

\tightparagraph{Theorem~\ref{thm:equivalent}}
\emph{Failure detector classes $\Sigma$, $\HS$, and $\AS$ are equivalent in $AS[\emptyset]$. Furthermore,
the transformation between $\Sigma$ and $\HS$ do not require initial knowledge of the membership.}

\tightparagraph{Proof of Theorem~\ref{thm:equivalent}}
From Theorems~\ref{HS-Sigma} and~\ref{Sigma-HS} we have that $\Sigma$ and $\HS$ are equivalent.  The equivalence between $\Sigma$ and $\AS$ was shown in \cite{DBLP:conf/wdag/BonnetR10}.

\subsubsection{From $\AS$ to $\HS$}

We show now how to obtain a failure detector of class $\HS$ from a detector of class $\AS$.

\tightparagraph{Theorem~\ref{thm:astohs}}
\emph{Class $\HS$ can be obtained from class $\AS$ in $AAS[\emptyset]$ without communication.}

\tightparagraph{Proof of Theorem~\ref{thm:astohs}}
Let $D$ be a detector of class $\AS$.
The transformation can be done as follows. Let $\bot$ be the ``default" identifier. Let us denote with $\bot^r$ a multiset of $r$ identifiers $\bot$. Each process $p$ periodically does as follows. For each pair $(x,y) \in D.a\_sigma_p$, the label $x$ is included in $\labels_p$ and the pair $(x, \bot^y)$ is included in $\quora_p$ (replacing any pair $(x,-)$ that $\quora_p$ may contain). The properties of $\HS$ follow trivially from the properties of $\AS$.

\subsubsection{From $\NAP$ to $\HP$ and $\HS$}

We show here how failure detectors of the classes $\HP$ and $\HS$ can be obtained for a failure detector of class $\NAP$ without communication.

\begin{lemma}
A failure detector of class $\HP$ can be obtained from any detector $D$ of class $\NAP$ in $AAS[\emptyset]$ (an anonymous asynchronous system) without communication.
\end{lemma}
\begin{proof}
The transformation can be done as follows. Let $\bot$ be the ``default" identifier. Each process $p$ 
%sets $\leader_p$ to the value $\bot$ initially and never changes it. Then, it updates 
periodically updates $h\_trusted_p$ to a multiset of $D.anap_p$ identifiers $\bot$. The liveness property of $D$ guarantees the liveness property of $\HP$.
\end{proof}

\begin{lemma}
A failure detector of class $\HS$ can be obtained from any detector $D$ of class $\NAP$ in $AAS[\emptyset]$ (an anonymous asynchronous system) without communication.
\end{lemma}
\begin{proof}
The transformation can be done as follows. Let $\bot$ be the ``default" identifier. Let us denote with $\bot^r$ a multiset of $r$ identifiers $\bot$. 
Each process $p$ periodically does as follows. After obtaining a value $y$ from $D.anap_p$, the label $\bot^y$ is included in $\labels_p$ and the pair $(\bot^y, \bot^y)$ is included in $\quora_p$. The Validity and Monotonicity of $\HS$ hold trivially. Liveness follows since, from the safety of $\NAP$, only correct processes see an output of $D.anap=c=|\C|$, and from the liveness property all of them do it. Then, every correct process $p$ eventually inserts $\bot^c$ in $\labels_p$ and $(\bot^c, \bot^c)$ in $\quora_p$, and only those processes. Safety of $\HS$ comes from the safety property of $\NAP$: if, for any $y$ and $y'$ with $y \geq y'$, $|S(\bot^y)| = y$ and $|S(\bot^{y'})|=y'$ (none can be larger), then $S(\bot^y) \subseteq S(\bot^{y'})$.
\end{proof}

\tightparagraph{Theorem~\ref{thm:nap}}
\emph{Classes $\HP$ and $\HS$ can be obtained from class $\NAP$ in $AAS[\emptyset]$ without communication.}

\tightparagraph{Proof of Theorem~\ref{thm:nap}}
The proof of Theorem~\ref{thm:nap} follows from the two previous lemmas.